\documentclass[11pt,a4paper]{article}
\usepackage[english]{babel} 
\usepackage{amsmath}
\usepackage{amsfonts}
\usepackage{amssymb}
\usepackage{amsthm}
\usepackage{makeidx}
\usepackage{graphicx} 
\usepackage{natbib}
\usepackage{setspace}
\usepackage{float}
\usepackage{multirow}

\usepackage{hyperref}
\usepackage[noabbrev]{cleveref}
\usepackage[margin=0.66in]{geometry}
\usepackage{lscape} 
\usepackage{setspace}
\usepackage[colorinlistoftodos]{todonotes} 
\hypersetup{
	colorlinks   = true,
	citecolor    = blue,
	urlcolor = blue
}

\newcommand{\E}{{E}}
\newcommand{\R}{\mathbb{R}}

\DeclareMathOperator*{\argmax}{arg\,max}
\DeclareMathOperator*{\argmin}{arg\,min}
\DeclareMathOperator*{\essinf}{ess\,inf}
\DeclareMathOperator*{\esssup}{ess\,sup}

\usepackage{tabularx}

\usepackage{ragged2e}
\title{ A risk measurement approach from risk-averse stochastic optimization of score functions}

\author{Marcelo Brutti Righi$^{a,}$\footnote{Corresponding author.}\\  \small{ \href{mailto:marcelo.righi@ufrgs.br}{marcelo.righi@ufrgs.br}} \and  Fernanda Maria M\"{u}ller$^{a}$\\\small{\href{mailto:fernanda.muller@ufrgs.br}{fernanda.muller@ufrgs.br}}
\and Marlon Ruoso Moresco$^{b}$\\\small{\href{mailto:marlon.moresco@concordia.ca}{marlon.moresco@ufrgs.br}}}

\date{\small{$^{a}$\textit{Business School, Federal University of Rio Grande do Sul, Washington Luiz, 855, Porto Alegre, Brazil, zip 90010-460}\\
$^{b}$\textit{Department of mathematics and statistics, Concordia University, Boul. de Maisonneuve Ouest, 1455, Montréal, QC, Canada, zip H3G 1M8}}}

\newtheorem{Def}{Definition}
\newtheorem{Thm}{Theorem}
\newtheorem{Prp}{Proposition}
\theoremstyle{definition}
\newtheorem{Exm}{Example}

\theoremstyle{remark}
\newtheorem{Rmk}{Remark}

\begin{document}
	
	\maketitle
	
	\begin{abstract}
We propose a risk measurement approach for a risk-averse stochastic problem.
We provide results that guarantee that our problem has a solution. We characterize and explore the properties of the argmin as a risk measure and the minimum as a deviation measure. 
We provide a connection between linear regression models and our framework.
Based on this conception, we consider conditional risk and provide a connection between the minimum deviation portfolio and linear regression. 
 Moreover, we also link the optimal replication hedging to our framework.
 
		\textbf{Keywords}:  Risk management; Uncertainty modeling; risk measures; deviation measures; robust stochastic programming.
	\end{abstract}
	
\section{Introduction}\label{introduction}

The theoretical discussion of risk measures  gained prominence since the seminal work of \cite{Artzner1999}, who developed the class of coherent risk measures.  
From there, other proprieties and classes of risk measures were proposed, including the convex \citep{Follmer2002, Frittelli2002},  spectral \citep{Acerbi2002}, and generalized deviation measures  \citep{rockafellar2006generalized}. From that, an entire stream of literature has proposed and discussed distinct features for risk measures, including axiom sets, dual representations, and mathematical  properties. For detailed reviews, we recommended the books of \cite{Pflug2007}, \cite{Delbaen2012}, \cite{Ruschendor2013}, and \cite{Follmer2016} and the studies  of  \cite{Follmer2013} and \cite{Follmer2015}.

Recently, the discussion of statistical properties that a risk measure must respect has also gained space in the literature that discusses characteristics for risk measures.
A prominent statistical property is elicitability. This property is very useful for risk management because it enables comparing competing forecast models using the scoring rule.  Examples of elicitable functionals are quantiles and expectiles, which makes Value at Risk (VaR) and Expectile Value at Risk (EVaR) elicitable risk measures. 
 We recommended \cite{gneiting2011making}, \cite{Bellini2015}, \cite{Ziegel2016},  \cite{Kou2016}, \cite{fissler2016higher}, \cite{fissler2021}, and the references therein for more details. 
 A functional $T$ on a vector space of random variables as $X$ is elicitable if exists a scoring function   $S\colon\mathbb{R}^{2}\rightarrow\mathbb{R}_+$ such that
	\begin{equation*}
	T(X)=-\argmin\limits_{y\in\mathbb{R}}E\left[S(X,y) \right].
	\end{equation*}
	We present more details in  \Cref{def:eli} below, and \Cref{Examples} describe some examples of $S$.

Inspired on the elicitability reasoning, even without keeping its technical definition, we consider a robust/risk-averse counterpart to the optimization problem as {\[\inf_{y\in \R} \sup\limits_{\mathbb{Q}\in\mathcal{Q}^\prime} E_\mathbb{Q}[S(X,y) ],\]}where $\mathcal{Q}^\prime$ is a suitable set of probability measures, which may represent beliefs or scenarios. Regarding the ambiguity set, one can choose $\mathcal{Q}^\prime\subseteq\mathcal{P}$ 
		in an ad hoc sense according to some \textit{a priori} established risk aversion parameter. Another possibility is to consider those probability measures representing beliefs absolutely continuous inside some distance from a nominal measure $\mathbb{P}$, as in \cite{Shapiro2017}. We consider $\mathcal{Q}^\prime$ linked to dual representations of coherent risk measures (sub-linear expectations as in \cite{Sun2017}). We consider risk measures coherent in the sense of \cite{Artzner1999} because these maps have a dual representation as the supremum of expectations over a closed (in total variation norm) convex set of probability measures. Thus, with coherent risk measures $\rho$ replacing the expectation, we define our risk measurement approach as a risk-averse stochastic problem as
\[\inf_{y\in \R}  \rho\left(-S(X,y) \right). \]

A possible, but not limited to, interesting direct application of this kind of risk measurement process could seek to minimize capital
determination errors to reduce the costs linked to it. As pointed out in \cite{Righi2020}, from the
regulatory point of view, risk underestimation, and consequently
capital determination underestimation, is the main concern. In
this case, capital charges are desirable to avoid costs from unexpected and uncovered losses. However, from the perspective of
institutions, it is also desirable to reduce the regret costs arising
from risk overestimation because the latter reduces profitability.

We provide results that guarantee that our risk measurement approach has a solution. We characterize the argmin as a risk measure per se, and the minimum as a deviation measure in the sense of \cite{rockafellar2006generalized}. We also explore the main proprieties of both functionals. 
Our proposal is inspired by the study of \cite{Righi2020}. 
The authors  propose a risk measurement procedure that represents the capital determination for a financial position that minimizes the expected value of the sum between costs from risk overestimation and underestimation and considers a supremum of probability measures to the expectation. However, they explore a single score instead general ones. A similar idea has been pursued in \cite{Mao2018}, where the expectation has been
replaced by functionals arising from rank-dependent expected utility and cumulative prospect theory. In this way, both studies of can be thought of as special cases in our framework.

The paper of \cite{Rockafellar2013} also relates to risk and deviation measures linked by a common optimization problem, and 
 \cite{Bellini2014}  study generalized quantiles as risk measures by minimizing asymmetric loss functions.
Unlike we do in this current approach, both mentioned studies do not consider the supremum of probability measures to the expectation. 
Thus, our approach can be considered robust since it is not sensitive to choosing a specific probability measure representing a particular belief about the world. In this sense, our approach is in concordance with the stream of \cite{Shapiro2017}, \cite{Bellini2019}, \cite{Righi2018b}, \cite{Righi2020}, for instance. 
In line with our study,  \cite{Embrechts2021} introduce the notions of Bayes pairs and Bayes risk measures as the counterpart of elicitable risk measures as the minimum of the scores. Nonetheless, their minimum scores are also risk measures instead of deviations.

We also make a connection between our framework and linear regression analysis. The most common functional form of regression analysis is linear regression, widely known through the method of ordinary least squares that minimizes the sum of squared differences. Other forms of regression use slightly different scores to estimate parameters, such as the quantile regression, see \cite{Koenker1978}, \cite{koenker2005quantile}, expectile regression, see \cite{Newey1987}, and extremile regression, see \cite{Daouia2019} and \cite{Daouia2021}, for instance.
The link between linear regression models and our risk measurement approach allows us to have conditional versions of both risk and deviation.
We explore the proprieties of conditional risk and prove that the minimizer is unique.
Discussions of conditional versions of risk are not new in the literature. 
However, the focus has been on score functions related mainly to quantile regression, i.e., VaR regressions. \cite{guillen2021joint} point that this approach is extremely useful for identifying covariates that influence the worst-case outcomes. 
We extend this discussion to different score functions. \cite{Wu2023} explores, as a counterpart to the generalized quantiles studied in \cite{Bellini2014}, conditional generalized quantiles. They, contrary to us, do not consider a robust optimization approach.

The concept of deviation is present in finance since \cite{markowitz1952portfolio} with the standard deviation. Such concept is axiomatized and generalized for convex functionals in \cite{rockafellar2006generalized}, \cite{pflug2006subdifferential} and \cite{Grechuk2009}. The problem of minimizing the deviation of a portfolio and its implications are explored in \cite{Rockafellar2007}. Recently, \cite{Righi2016} and \cite{Righi2018d} consider both risk and deviation measures. Furthermore, representing the portfolio choice problem in terms of an estimation problem of a linear
regression model is well known. \cite{Britten1999} proposes a regression approach for the
tangency portfolio, and \cite{Kempf2006} as well as \cite{Fan2012} show that the
plug-in estimator for the GMVP (global minimum variance portfolio) weights can also be obtained by means of linear regression.
More recently, \cite{Li2015} provides a regression representation of the mean-variance portfolio. The approach in \cite{Frey2016}
differs from the regression representation mentioned above by  avoiding the choice for a n-th asset $X_n$ as a dependent variable. 
We provide a similar connection between minimum deviation problems with linear regression under the same score that generates the deviation. Our results guarantee that our minimum deviation portfolio optimization problem has a solution.

In a complete market model, any derivative is attainable and thus admits a perfect hedge. The cost of replication equals the price of the derivative, which is the expected discounted claim payoff under the unique equivalent martingale measure \citep{huang2013optimal}. However, completeness is only an idealization of a financial market. Relaxing the idealized assumption leads to incomplete market models, where financial products bear an intrinsic risk that cannot be hedged away completely, see \cite{carr2001pricing}, and \cite{balter2020pricing} for details. For hedging procedures proposed in the literature for expected utility maximization in the form of minimization of a score/loss function, typically quadratic or quantile one, see \cite{bessler2016analyzing}, \cite{halkos2019energy} and \cite{barigou2022insurance}, for instance.  We then provide a direct connection between optimal hedging strategies with linear regression under the same score that the hedge is taken. We explore results that guarantee that our problem has a solution. 


Regarding structure, the remainder of this paper divides in the
following contents:  \Cref{Preli} describes definitions and results concerning the existence of a solution to our risk measurement approach problem, and explores the properties of our risk and deviation measures.  \Cref{Examples} exposes, in more detail,  examples of possible choices for $\rho$ and $S$.  \Cref{CondRisk} connects our approach to linear regression models, allowing conditional risk and its properties, besides solving minimum deviation portfolio optimization and optimal replication hedging problems. 

\section{Proposed approach} \label{Preli}

Consider the real-valued random result $X$ of any asset ($X\geq0$ is a gain, $X<0$ is a loss) that is defined on a probability space $(\Omega,\mathcal{F},\mathbb{P})$. All equalities and inequalities are considered almost surely in $\mathbb{P}$. We define $X^+=\max(X,0)$, $X^-=\max(-X,0)$, and $1_A$ as the indicator function for an event  $A$. Let $L^{p}:=L^{p}(\Omega,\mathcal{F},\mathbb{Q})$ the space of (equivalent classes of) random variables such that $ \lVert X \rVert_p^p = E[|X|^p]<\infty$ for $p\in[1,\infty)$ and $ \lVert X \rVert_\infty = \esssup |X| < \infty$ for $ p = \infty$, where $E$ is the expectation. When not explicit, it means that definitions and claims are valid for any fixed $L^p,\:p\in[1,\infty]$. We have that $L^p_{+}$ is its cone of non-negative elements. We denote by $X_n\rightarrow X$ convergence in the $L^p$ norm, while $\lim\limits_{n\rightarrow\infty}X_n=X$ means $\mathbb{P}$-a.s. convergence. 

We let $\mathcal{Q}$ denote the set composed of probability measures  $\mathbb{Q}$ defined on $(\Omega,\mathcal{F})$ that are absolutely continuous with respect to $\mathbb{P}$, with Radon-Nikodym derivative $\frac{d\mathbb{Q}}{d\mathbb{P}}\in L^q$, $\frac{1}{p}+\frac{1}{q}=1$, with the convention $q=\infty$ when $p=1$ and $q=1$ when $p=\infty$.
  Moreover, $E_{\mathbb{Q}}[X]=\int_{\Omega}Xd\mathbb{Q}$, $F_{X, \mathbb{Q}}(x)=\mathbb{Q}(X\leq x)$ and $F_{X, \mathbb{Q}}^{-1}(\alpha)=\inf\{x\in\mathbb{R}\colon F_{X, \mathbb{Q}}(x)\geq\alpha\}$ are, respectively, the expected value, the distribution function and the (left) quantile of $X$ under $\mathbb{Q}$. We drop the subscript  when it is regarding  $\mathbb{P}$.

We now formally define the framework we need to build our proposed approach. 

\begin{Def}\label{def:eli}
	A map $S:\mathbb{R}^{2}\rightarrow\mathbb{R}_+$ is called scoring function if the map $\omega\rightarrow S(X(\omega),Y(\omega))$  belongs to $L^1$ for any $X,Y\in L^p$, and satisfy the following properties for any $x,y\in\R$:
	
	\begin{enumerate}
		\item $S(x,y)\geq 0$ and $S(x,y)=0$ if and only if $x=y$.
		\item There is a function $f_S\colon\mathbb{R}\to\mathbb{R}$ such that $S(x,y)=f_S(x-y)$.
		\item $y\to S(x,y)$ is convex and continuous. 
	\end{enumerate}
	A function $T\colon L^p\rightarrow\mathbb{R}$ is elicitable if exists a scoring function $S$ such that
	\begin{equation}\label{eq:eli}
T(X)=-\argmin\limits_{y\in\mathbb{R}}E\left[S(X,y) \right],\:\forall\:X\in L^p.
	\end{equation}
\end{Def}

\begin{Rmk}
As a consequence of properties (i) and  (iii) we have that $y\to S(x,y)$ is
 non-decreasing for $y>x$ and non-increasing for $y<x$. Moreover, some more generality can be obtained. In fact, for most of the paper one could relax continuity of $y\to S(x,y)$ to only lower semi-continuity. Further, one can drop the demand for existence of a $f_S$ at the cost of dropping the Translation Invariance/Insensitivity (see below). It is straightforward to verify that there is the preservation of such properties if and only if there exists such a real $f_S$. 
\end{Rmk}

\begin{Rmk} 
We would like to highlight that the assumption on the scoring function $S$ implies some properties in the function $f_S$. In particular, we have that $f_S (x) =0$ if and only if $x = 0$, $f_S(x)$ is convex, continuous, non-decreasing for $x<0$ and non-increasing for $x>0$. This also implies that the map $x \mapsto S(x,y)$ has the same properties as $y\to S(x,y)$. Furthermore, note that when the necessary derivatives exist, we have that $\frac{\partial S(x,y)}{\partial x} = \frac{\partial f_S(x-y)}{\partial x} = f_S'(x-y)$ and $\frac{\partial S(x,y)}{\partial y} = \frac{\partial f_S(x-y)}{\partial y} = -f_S'(x-y)$.
\end{Rmk}

 A robust counterpart to this optimization problem, even without keeping the technical definition of elicitability, involves sets of probability measures obtained from coherent risk measures. Thus, we expose some definitions and results from the risk measures literature we use alongside the paper. We choose to consider only finite maps since it is the kind that fits our proposed approach.

\begin{Def}\label{def:risk}
	A functional $\rho:L^p\rightarrow\mathbb{R}$ is a risk measure. Its acceptance set is defined as $\mathcal{A}_\rho=\left\lbrace X\in L^p:\rho(X)\leq 0 \right\rbrace $. $\rho$ may possess the following properties: 
	
	\begin{enumerate}
		\item Monotonicity: if $X \leq Y$, then $\rho(X) \geq \rho(Y),\:\forall\: X,Y\in L^p$.
		\item Translation Invariance: $\rho(X+c)=\rho(X)-c,\:\forall\: X\in L^p,\:\forall\:c \in \mathbb{R}$.
		\item Convexity: $\rho(\lambda X+(1-\lambda)Y)\leq \lambda \rho(X)+(1-\lambda)\rho(Y),\:\forall\: X,Y\in L^p,\:\forall\:\lambda\in[0,1]$.
		\item Positive Homogeneity: $\rho(\lambda X)=\lambda \rho(X),\:\forall\: X\in L^p,\:\forall\:\lambda \geq 0$.

	\end{enumerate}
	
	We have that  $\rho$ is called monetary if it fulfills (i) and (ii), convex if it is monetary and respects (iii), and coherent if it is convex and fulfills (iv).  
\end{Def}

\begin{Thm}[Theorems 2.11 and 3.1 of \cite{Kaina2009}]\label{the:dual}
	A map $\rho : L^p\rightarrow \mathbb{R}$, $p\in[1,\infty)$, is a coherent risk measure if and only if it can be represented as:
	\begin{equation}\label{eq:cohdual}
	\rho(X)=\max\limits_{\mathbb{Q}\in\mathcal{Q}_\rho} E_\mathbb{Q}[-X],\:\forall\:X\in L^p,
	\end{equation} where $\mathcal{Q}_\rho\subseteq\mathcal{Q}$ is non-empty, closed, and convex set called the dual set of $\rho$. Moreover, $\rho$ is lower semi-continuous in the $L^p$ norm and continuous in the bounded $\mathbb{P}$-a.s. convergence (Lebesgue continuous).
\end{Thm}

Thus, we can have risk measures replacing the expectation under an appropriate choice for the dual set. Hence, we define our risk measurement approach as a risk-averse stochastic problem as
\[\inf_{y\in \R}  \rho\left(-S(X,y) \right). \]
In order to guarantee risk-averseness of the problem in the sense of worst values for the objective function, we assume that $\mathbb{P}\in\mathcal{Q}_\rho$, which implies $\rho(X)\geq E[-X]$ for any $X\in L^p$. This property is known as Loadedness in the literature. This is always the case when $\rho$ is law invariant in the sense that $F_X=F_Y$ implies $\rho(X)=\rho(Y)$, see \cite{Follmer2016} for details.

We now formally define the functionals from our proposed risk measurement approach. The negative sign is to keep the pattern for losses. Moreover, in Propositions \ref{prp:sol2} and \ref{prp:sol} we prove that our both functionals are in fact finite and, thus, well defined. 

\begin{Def}\label{def:meas} Let $\rho\colon L^1\to\mathbb{R}$ be a coherent risk measure and $S\colon \mathbb{R}^2\to\mathbb{R}$ a score function. The risk and deviation measures they generate are functionals $R,D\colon L^p\to\mathbb{R}$, $p\in[1,\infty)$, respectively, defined as
	\begin{equation}
	R(X):=R_{\rho,S}(X)= -\min\left\lbrace\argmin\limits_{y\in\mathbb{R}} \rho\left( -S(X,y)\right) \right\rbrace \label{averserisk},
	\end{equation}
	and
	\begin{equation}
	D(X):=D_{\rho,S}(X)= \min\limits_{y\in\mathbb{R}} \rho\left( -S(X,y)\right)  \label{aversedev}.
	\end{equation}
Furthermore, for any $X \in L^p$ define the set of minimizers as $B_X:=\argmin\limits_{y\in\mathbb{R}} \rho\left( -S(X,y)\right) $.
\end{Def}

\begin{Rmk}
In this definition, we consider $L^1$ as the domain for $\rho$ to be a  more general and easy notation. Nonetheless, any $L^r,r\in[1,\infty)$ could be considered by adjusting the definition of score $S$ to fulfills $S(X,Y)\in L^r$ for any $X,Y\in L^p$. All results in this paper are directly adaptable to the $\rho\colon L^r\to\mathbb{R}$ if that would be the case. We just do not consider $L^\infty$ for the domain of $\rho$ since we want the supremum in its dual representation to be attained. In $L^\infty$, this would be the case under further continuity properties. See \cite{Follmer2016} for details. We do not want to consider parsimony. Similar reasoning applies to  both $R$ and $D$.
\end{Rmk}

\begin{Rmk}
	The task of assessing the performance of financial investments is central, with indexes such as the Sharpe ratio used to assess the trade-off between risk and return. In the last decade, performance has been analyzed through acceptability indexes since the seminal paper of \cite{Cherny2009}, which is extended in \cite{Gianin2013}. These authors discuss the properties such functionals must fulfill. Under our framework, we can have a reward to deviation ratio for acceptability as a map $RD\colon L^p\to [0,\infty]$ defined as 
		\begin{equation*}
	RD(X)=
		\begin{cases}
	\dfrac{-R(Y)}{D(Y)}&\text{if} \:R(X)<0\:\text{and}\:D(X)>0,\\
				\infty&\text{if}\:R(X)\leq 0 \:\text{and}\:D(X)=0,
				\\ 0&\text{otherwise}.		\end{cases}	
		\end{equation*} 
	Other adaptations and properties of this structure are discussed in and \cite{Righi2021b}.
\end{Rmk}

We now expose a formal result that guarantees our minimization problems have a solution. 

\begin{Prp}\label{prp:prob}
	Let $X\in L^p$ and $B_X$ defined as in \Cref{def:meas}. Then:
	\begin{enumerate}
		\item  $B_X$ is a closed interval.
		\item  $y\in B_X$ if and only if $y$ satisfies the first order condition given by 
		\begin{equation}
		E_{\mathbb{Q}^{*}}\left[\dfrac{\partial^{-}S(X,x)}{\partial x}(y) \right]\leq 0 \leq  	E_{\mathbb{Q}^{*}}\left[\dfrac{\partial^{+}S(X,x)}{\partial x}(y)\right],
		\end{equation} where $\mathbb{Q}^{*}=\argmax\limits_{\mathbb{Q}\in\mathcal{Q}_\rho} E_\mathbb{Q}[S(X,y)]$.
		\item if $y\to \rho(- S(X,y))$ is, for any $X\in L^p$, differentiable with strictly increasing derivative, then $B_X$ is a singleton. 
	\end{enumerate} 
\end{Prp}

\begin{proof}
	For (i), fix $X\in L^p$ and let $f:  \R \rightarrow \R$ be defined as  $f (x):=  \rho\left( -S(X,x)\right) $. Clearly, $f$ is finite, convex, and, hence, a continuous function.  Note that  $f$ is  proper and level bounded. Thus,  $\inf_{x\in\mathbb{R}}  f(x)$ is finite and the set $\argmin_{x\in\mathbb{R}}  f(x)$ is non-empty and compact. Moreover, since $f$ is convex, $B_X$ is an interval. 
	
	Regarding (ii),	since $ f$ is convex, we have that $y\in\mathbb{R}$ is a minimizer if and only if
	\[0\in\left[\dfrac{\partial^-  f(x)}{\partial x}(y),\dfrac{\partial^+  f(x)}{\partial x}(y)\right].\] 
	Dominated convergence yields  \[
	\dfrac{\partial^- f(x)}{\partial x}(y)=E_{\mathbb{Q}^{*}}\left[\dfrac{\partial^{-}S(X,x)}{\partial x}(y) \right]\:\text{and}\:\dfrac{\partial^+ f(x)}{\partial x}(y)=E_{\mathbb{Q}^{*}}\left[\dfrac{\partial^{+}S(X,x)}{\partial x}(y)\right],\]
	where $\mathbb{Q}^* \in \mathcal{Q}_\rho$ such that $\rho(-S(X,y)) = E_{\mathbb{Q}^*} [S(X,y)]$.
	
	For (iii) the f.o.c. becomes \[E_{\mathbb{Q}^{*}}\left[\dfrac{\partial S(X,x)}{\partial x}(y)\right]=\dfrac{\partial \rho(-S(X,x))}{\partial x}(X,y)=0.\] Then, the strictly increasing behavior assures the minimizer is unique.
\end{proof}

We now explore the main properties of our risk and deviation measures.

\begin{Prp}\label{prp:sol2}
	Let $R$ and $B_X$ be as in \Cref{def:meas}. Then:
	\begin{enumerate}
		\item $R$ is monetary and $\mathcal{A}_{R}= \left\lbrace X \in L^p:  y\not\in B_X \; \forall \; y<0\right\rbrace$.
		\item  if  $y\to \rho(- S(X,y))$ is, for any $ X\in L^p$, differentiable with strictly increasing derivative and $\frac{\partial f_S(x)}{\partial x}$ is convex, then $R$ fulfills Convexity. In this case, $R$ is a lower semi-continuous in the $L^p$ norm and continuous in the bounded $\mathbb{P}$-a.s. convergence (Lebesgue continuous).
		\item if $f_S$ is positive homogeneous, then $R$ fulfills Positive Homogeneity.
		\item  $R(X)\in[\essinf X,\esssup X]$.
	\end{enumerate}
	
\end{Prp}

\begin{proof}
	Regarding (i), Translation Invariance is straightforward since $S(x,y)=f_S(x-y)$ with $f_S(0)=0$. For Monotonicity, let $g,h\colon L^p\times\mathbb{R}\rightarrow\mathbb{R}$ be as  
	\[g(X,x)=\frac{\partial^{-}\rho\left( -S(X,y)\right)}{\partial y}(X,x),\:\text{and}\:h(X,x)=\frac{\partial^{+}\rho\left( -S(X,y)\right)}{\partial y}(X,x).\]
	Since $x \mapsto \rho(-S(X,x ) )$ is a convex real function for all $X \in L^p$, the left and right derivatives above are well defined. Furthermore, $g$ and $h$ are non-decreasing in the second argument. Additionally, note that
	\begin{align*}
	g(X,x) = E_{\mathbb{Q}^*} \left[ \frac{\partial^{-} S(X,y) }{\partial y}(X,x) \right].
	\end{align*}
	Hence, we have that $X(\omega) \mapsto S(X(\omega),x) $ is also a convex real function. Therefore,  $X\to g(X,x)$ is non-increasing. Similarly for $h$.  Now, let $X,Y\in L^p$ with $X\leq Y$. Then $h(Y,x)\leq h(X,x)$ for any $x\in\mathbb{R}$. Furthermore, as $x \mapsto g(X,x)$ is non-decreasing and $g \leq h$, the condition $g(X,x) \leq 0 $ in the following is non-binding in the sense that
	\[ \inf\{x\in\mathbb{R}\colon g(X,x)\leq 0,h(X,x)\geq0\} = \inf\{x\in\mathbb{R}\colon h(X,x)\geq0\}.\]
	Then, we get from the first order condition of  \Cref{prp:prob} that		\begin{align*}
	-R(X) &=\min B_X=\inf\{x\in\mathbb{R}\colon h(X,x)\geq0\},\\
	-R(Y)&=\min B_Y=\inf\{x\in\mathbb{R}\colon h(Y,x)\geq 0\}.
	\end{align*} Note that such expressions are well defined because, from  \Cref{prp:prob}, the argmin set is a closed interval. We then must have $-R(X)\leq -R(Y)$ since $\{x\in\mathbb{R}\colon h(Y,x)\geq 0\}\subseteq \{x\in\mathbb{R}\colon h(X,x)\geq 0\}$. By multiplying both sides by $-1$ we get the claim. Moreover, we then have that
	\[\mathcal{A}_{R}= \left\lbrace X \in L^p:\min B_X\geq 0 \right\rbrace =\left\lbrace X \in L^p: y\in B_X \Rightarrow y \geq 0\right\rbrace=\left\lbrace X \in L^p:  y\not\in B_X, \forall \; y<0\right\rbrace.\]
	
	Concerning (ii), let $f_S^\prime$ be convex. The f.o.c. becomes \[E_{\mathbb{Q}^{*}}\left[f_S^\prime(X-x)\right]=\dfrac{\partial \rho(-f_S(X-x))}{\partial x}(X-x)=0.\]
	Let then $g\colon L^p\times\mathbb{R}\to\mathbb{R}$ be as 
	\[g(X,x)=\frac{\partial\rho\left( -f_S(X-y)\right)}{\partial y}(x),\] which is convex in its domain and non-increasing in $x$ for any $X\in L^p$. Let $\lambda\in[0,1]$ and $X,Y\in L^p$. Then we have 
	\begin{align*}
	g \left( \lambda X+(1-\lambda)Y,\lambda R(X)+(1-\lambda)R(Y) \right) 
	\leq  \lambda g(X,R(X))+(1-\lambda)g(Y,R(Y))=0.
	\end{align*} 
Furthermore, $0= g(\lambda X+(1-\lambda)Y,R(\lambda X+(1-\lambda)Y)).$ This yields 
\[g \left( \lambda X+(1-\lambda)Y,\lambda R(X)+(1-\lambda)R(Y) \right)  \leq g(\lambda X+(1-\lambda)Y,R(\lambda X+(1-\lambda)Y)).\]	
	Thus, due to its non-increasing behavior in the second argument, we obtain $R(\lambda X+(1-\lambda)Y)\leq\lambda R(X)+(1-\lambda)R(Y))$. In this case, $R$ is a convex risk measure. The continuity properties are then directly obtained from \Cref{the:dual}, in fact, from Theorems 2.11, and 3.1 of \cite{Kaina2009}.
	
	Regarding (iii), the result follows immediately since  for any $X\in L^p$, any $y\in\mathbb{R}$ and $\lambda\geq0$ \begin{align*}
	R(\lambda X) &=
	-\min\left\lbrace\argmin\limits_{y\in\mathbb{R}} \rho\left( -f_S(\lambda X - y)\right) \right\rbrace \\
	&=  -\min\left\lbrace\argmin\limits_{y\in\mathbb{R}} \lambda \rho\left( -f_S(X - \frac{y}{\lambda} )\right) \right\rbrace \\ &=
	- \lambda \min\left\lbrace\argmin\limits_{y\in\mathbb{R}} \rho\left( -f_S( X - y)\right) \right\rbrace = \lambda R(X).
	\end{align*}

	For (iv), for $X \in L^\infty$ note that $f_S(X-y)\geq f_S(X-\operatorname{ess}\inf X)$ for $y<\operatorname{ess}\inf X$, and $f_S(X-y)\geq f_S(X-\operatorname{ess}\sup X)$ for $y\geq\operatorname{ess}\sup X$. Additionally, when $\esssup X = \infty$ or $\essinf X = -\infty$, in other words, when $X \in L^p \setminus L^\infty$,  the condition becomes trivial as $R(X)$ is finite. Thus, the argmin must be in the interval. 
\end{proof}

\begin{Rmk}
	Under the conditions of the items (ii) and (iii), we have by  \Cref{the:dual}  the following dual representation:
	\[R(X)=\max\limits_{\mathbb{Q}\in\mathcal {Q}_R}E_\mathbb{Q}[-X],\:\forall\:X\in L^p,\] 
	where \begin{align*}
	\mathcal {Q}_R&=\left\lbrace \mathbb{Q}\in\mathcal{Q}\colon E_\mathbb{Q}[-X]\leq R(X),\:\forall\:X\in L^p\right\rbrace\\
	&=\left\lbrace \mathbb{Q}\in\mathcal{Q}\colon E_\mathbb{Q}[X]\geq \min B_X,\:\forall\:X\in L^p\right\rbrace
	\\ &=\bigcap_{X \in L^p}\left\lbrace \mathbb{Q}\in\mathcal{Q}\colon E_\mathbb{Q}[X]\geq \min B_X\right\rbrace. 
	\end{align*}
\end{Rmk}

We now characterize the minimum $D_{\rho,S}$ as a deviation measure. In this sense, we first formally define deviation measures.

\begin{Def}\label{def:dev}
	A functional $\mathcal{D}:L^p\rightarrow\mathbb{R}_+$ is a deviation measure. It may fulfill the following properties:
	\begin{enumerate}
		\item Non-Negativity: $\mathcal{D}(X)=0$ for $X \in \R$ and $\mathcal{D}(X)>0$ for
		 $X\in L^p \setminus \R$;
		\item Translation Insensitivity: $\mathcal{D}(X+c)=\mathcal{D}(X),\:\forall \:X\in L^p,\:\forall\:c \in\mathbb{R}$;
		\item Convexity: $\mathcal{D}(\lambda X+(1-\lambda)Y)\leq \lambda \mathcal{D}(X)+(1-\lambda)\mathcal{D}(Y),\:\forall\: X,Y\in L^p,\:\forall\:\lambda\in[0,1]$;
		\item Positive Homogeneity: $\mathcal{D}(\lambda X)=\lambda \mathcal{D}(X),\:\forall\:X\in L^p,\:\forall\:\lambda \geq 0$;
	\end{enumerate}
	A deviation measure $\mathcal{D}$ is called convex if it fulfills (i), (ii), and (iii); generalized (also called coherent) if it is convex and fulfills (iv).
\end{Def}

\begin{Prp}\label{prp:sol}
	Let $D$ be defined as in \Cref{def:meas}. Then it has the following properties:
	\begin{enumerate}	
		\item $D$ is a convex deviation. Moreover, $D(X)=\sup\limits_{\mathbb{Q}\in\mathcal{Q}_\rho}\min\limits_{y\in\mathbb{R}}E_\mathbb{Q}[S(X,y)],\:\forall\:X\in L^p$.
		\item if $f_S$ is positive homogeneous, then $D$ fulfills Positive Homogeneity. 
		\item if $f_S$ is sub-additive and $f_S(X)\leq\lVert X\rVert_p\:\forall\:X\in L^p$, then $D$ lower semi-continuous in the $L^p$ norm.
		\item $D$ is continuous in the bounded $\mathbb{P}$-a.s. convergence (Lebesgue continuous).
		\item if $\rho_1\geq \rho_2$ ($S_1\geq S_2$), then $D_{\rho_1,S}\geq D_{\rho_2,S}$ ($D_{\rho,S_1}\geq D_{\rho,S_2}$).
	\end{enumerate}
\end{Prp}

\begin{proof}
	Regarding (i), Translation insensitivity is direct from $D(X)=\rho(-f_S(X+R(X)))$. For Non-negativity, since $R(c)=-c,\:\forall\:c\in\mathbb{R}$, we have that $D(c)=0$. If $X$ is not a constant, with abuse of notation, we have that $\mathbb{P}(X\neq -R(X))>0$. We then get that $\mathbb{P}(S(X,R(X))>0)>0$, which, together to $S(X,R(X))\geq0$  guarantees that $E[S(X,R(X))] >0$. Hence, $D(X)=\rho(-S(X,R(X)))\geq E[S(X,R(X)) ] > 0$. 
	For convexity, remember that $f_S$ is convex. Then, consider any pair $X,Y\in L^p$ and any $\lambda\in[0,1]$. We then obtain  that
	\begin{align*}
	D(\lambda X+(1-\lambda)Y)&=\min\limits_{\lambda x_1+(1-\lambda)x_2\in\mathbb{R}}\rho\left(-f_S(\lambda X-\lambda R(X)+(1-\lambda)Y-(1-\lambda)R(Y)) \right)\\
	&\leq\min\limits_{x_1,x_2\in\mathbb{R}}\left\lbrace\lambda\rho(-f_S(X-R(X)))+(1-\lambda)\rho(-f_S(Y-R(Y))) \right\rbrace\\
	&=\lambda\min\limits_{x_1\in\mathbb{R}}\rho(-f_S(X-R(X)))+(1-\lambda)\min\limits_{x_2\in\mathbb{R}}\rho(-f_S(Y-R(Y)))\\
	&=\lambda D(X)+(1-\lambda)D(Y).
	\end{align*}	
		The representation  result follows from the Sion's minimax theorem, see \cite{Sion1958}, because the map $(x,\mathbb{Q})\rightarrow E_\mathbb{Q}[S(X,x)]$ has the needed continuity and quasi-convex properties, $\mathcal{Q}_\rho$ is convex. The optimization over $x\in\mathbb{R}$ can be done in the compact interval $B_X$. 
	
		Positive Homogeneity in (ii) is straightforwardly obtained when $f_S$ is positive homogeneous.	
	
	Regarding (iii), let $X_n\to X$. Since $D(X)\leq\rho(-S(X,y)),\:\forall\:y\in\mathbb{R}$, we thus get that $D(X)\leq\liminf\limits_{n\rightarrow\infty}\rho(-f_S(X+R(X_n)))$. We then have that \begin{align*}
	D(X)&\leq \liminf\limits_{n\rightarrow\infty}\rho(-f_S(X+R(X_n)))\\
	&= \liminf\limits_{n\rightarrow\infty}\rho(-f_S(X +X_n - X_n +R(X_n)))\\
	&\leq\liminf\limits_{n\rightarrow\infty}\left\lbrace\rho(-f_S(X-X_n))+\rho(-f_S(X_n+R(X_n))) \right\rbrace\\
	&\leq\liminf\limits_{n\rightarrow\infty}\left\lbrace\lVert X-X_n\rVert_p+\rho(-f_S(X_n+R(X_n))) \right\rbrace\\ 
	&=\liminf\limits_{n\rightarrow\infty}\lVert X-X_n\rVert_p+\liminf\limits_{n\rightarrow\infty}\rho(-f_S(X_n+R(X_n)))=\liminf\limits_{n\rightarrow\infty}D(X_n).
	\end{align*}
	
	For (iv),  for any $x\in\mathbb{R}$ and any $\mathbb{Q}\in\mathcal{Q}_\rho$ by Dominated convergence we have that $\lim\limits_{n\rightarrow\infty}X_n=X,\:\{X_n\}\subset L^\infty$ bounded implies, for any bounded sequence $\{x_n\}\subset\mathbb{R}$ such that $\lim\limits_{n\rightarrow\infty}x_n=x$, in
	\[E_{\mathbb{Q}}[S(X,x)]=\lim\limits_{n\rightarrow\infty}E_{\mathbb{Q}}[S(X_n,x_n)].\]
	Let $\{x^{*}_{\mathbb{Q},n}\}$ be a sequence where each member is from the argmin set, i.e., $x^*_{\mathbb{Q},n} \in \argmin\limits_{y\in\mathbb{R}}E_\mathbb{Q}\left[S(X_n,y) \right]$. Since $\{X_n\}$ is bounded, $\{x^{*}_{\mathbb{Q},n}\}$ also is bounded  for any $\mathbb{Q}\in\mathcal{Q}_\rho$. By the Bolzano-Weierstrass Theorem, we have, by taking a subsequence if needed, that $x^{*}_\mathbb{Q}=\lim\limits_{n\rightarrow\infty}x^{*}_{\mathbb{Q},n}$ is well defined and finite. We then get that \begin{align*}
	D(X)&\leq \sup\limits_{\mathbb{Q}\in\mathcal{Q}_\rho}E_{\mathbb{Q}}[S(X,x^{*}_\mathbb{Q})]\\
	&=\sup\limits_{\mathbb{Q}\in\mathcal{Q}_\rho}\lim\limits_{n\rightarrow\infty}E_{\mathbb{Q}}[S(X_n,x^{*}_{\mathbb{Q},n})]\\
	&\leq\liminf\limits_{n\rightarrow\infty}\sup\limits_{\mathbb{Q}\in\mathcal{Q}_\rho}E_{\mathbb{Q}}[S(X_n,x^{*}_{\mathbb{Q},n})]=\liminf\limits_{n\rightarrow\infty}D(X_n).
	\end{align*}
	Furthermore, if $\{X_n\}$ is bounded, then also is $\{S(X_n,x)\}$ for any $x\in\mathbb{R}$ since we have  $S(X_n,x)\leq\max\{S(M,x),S(-M,x)\}$, where $M$ is the uniform bound. By continuity of $S$ and Lebesgue continuity of $\rho$, we get that
	\begin{align*}
	D(X)&=\min\limits_{x\in\mathbb{R}}\lim\limits_{n\rightarrow\infty}\rho(-S(X_n,x))\geq\limsup\limits_{n\rightarrow\infty}\min\limits_{x\in\mathbb{R}}\rho(-S(X_n,x))=\limsup\limits_{n\rightarrow\infty}D(X_n).
	\end{align*}
	Hence $D(X)=\lim\limits_{n\rightarrow\infty}D(X_n)$.
	
	Finally, (v) is trivial from the Monotonicity of $\rho$, $\rho_1$ and $\rho_2$.
\end{proof}

\begin{Rmk}
	Under the conditions of item (iii), we have by Theorem 1 of  \cite{rockafellar2006generalized} and The Main Theorem of \cite{pflug2006subdifferential}, the following dual representation:
	\[D(X)=\sup\limits_{Z\in\mathcal {Z}_{\mathcal{Q}_\rho}}E[XZ],\:\forall\:X\in L^p,\] 
	where \begin{align*}
	\mathcal{Z}_{\mathcal{Q}^\prime}&=\left\lbrace Z\in L^q:E[Z]=0,E[XZ]\leq D(X),\:\forall\:X\in L^p\right\rbrace\\
	&=\left\lbrace Z\in L^q:E[Z]=0,\exists\:\:\mathbb{Q}\in\mathcal{Q}_\rho\:\text{s.t.}\:E[XZ]\leq E_{\mathbb{Q}}(S(X,x))\:\forall\:x\in\mathbb{R},\:\forall\:X\in L^p\right\rbrace\\
	&=clconv\left( \cup_{\mathbb{Q}\in\mathcal{Q}_\rho}\left\lbrace Z\in L^q:E[Z]=0,\:E[XZ]\leq  E_{\mathbb{Q}}[S(X,x)]\:\forall\:x\in\mathbb{R},\:\forall\:X\in L^p\right\rbrace\right),
	\end{align*} 
where $clconv$ means the closed convex hull.
\end{Rmk}

\begin{Rmk}
Recently, \cite{Castagnoli2021} proposed the class of star-shaped risk measures, which are characterized by the star-shaped property of the generated acceptance set. The reasoning for star-shapedness as a sensible axiomatic requirement is that if a position is acceptable, any scaled reduction is also possible.  This class is in the literature in
\cite{Liebrich2021}, \cite{Moresco2021}, \cite{Herdegen2021} and \cite{Righi2021b}. This property for some functional is defined as  $T(\lambda X) \leq \lambda T (X)$ for $\lambda \in[0,1]$. This  property is implied by convexity under $T(0)\leq 0$. In our framework, $R$ fulfills this property when we replace convexity of  $\frac{\partial f_S(x)}{\partial x}$ by star-shapedness and $\frac{\partial f_S(x)}{\partial x}(0)\leq 0$ in item (ii) of \Cref{prp:sol2}. For $D$, this property is automatically obtained since it is a convex deviation measure.
\end{Rmk}

\section{Examples}  \label{Examples}

In this section, we present a description of possible, but not all, choices for $\rho$ and $S$. The examples described for both quantities can be considered in the practical use of the proposed approach.

\subsection{Risk measures}\label{sub:risk}

In this subsection, we present examples of functionals that can be considered possible choices for $\rho$. However, it is noteworthy that the choices of $\rho$ are not limited to the risk measures presented.

\begin{Exm}
(Expected Loss). \label{ex:EL} Expected Loss (EL) is the most parsimonious coherent risk measure, and it indicates the expected value (mean) of a loss. Thus, EL is a functional $EL: L^1 \rightarrow \mathbb{R}$ defined as
\begin{align*}
EL(X)=E[-X]. 
\end{align*} 
For this measure, the dual set is a singleton $\mathcal{Q}_{EL}=\{\mathbb{P}\},$ that is, it only considers the basic belief. Henceforth, we will omit the subscript $E$ in $R_{E,S}$ and $D_{E,S}$ whenever the risk measure is the expected loss.
\end{Exm}

\begin{Exm}(Mean plus Semi-Deviation). The Mean plus Semi-Deviation (MSD) is a functional $MSD: L^2 \rightarrow \mathbb{R}$ defined by
\begin{align*}
MSD^{\beta}(X)=-E[X] + \beta \sqrt{E[((X-E[X])^{-})^2]}, 
\end{align*}
where $\beta \in [0,1]$. MSD penalizes the EL by the semi-deviation.  The proportion of deviation that has to be considered is given by $\beta$. This measure is studied in detail by \cite{ogryczak1999stochastic} and \cite{fischer2003risk}, and it is a well known law invariant coherent risk measure, which belongs to loss-deviation measures discussed by  \cite{righi2019composition}. The advantages of MSD are its simplicity and financial meaning.
 The dual set of this measure can be represented by $$\mathcal{Q}_{MSD^\beta}=\left \lbrace\mathbb{Q} \in \mathcal{Q} : \frac{d\mathbb{Q}}{d\mathbb{P}} =1+ \beta( V-E[V]), V\geq0, E[|V|^2]=1 \right\rbrace.$$
Despite not being finite for any $X\in L^1$, it is readily useful in our approach if we consider $L^2$ as the domain for $\rho$.
\end{Exm}

\begin{Exm} (Expected Shortfall). \label{ex:VaR} Expected Shortfall (ES) is the canonical example of a coherent risk measure, being the basis of many representation theorems in this field. Nowadays, it is recommended,  together with Value at Risk (VaR), by the Basel Committee as a functional basis for quantifying market risk. The ES is a functional $ES: L^1 \rightarrow \mathbb{R}$ defined as
\begin{align*}
ES^{\alpha}(X)= \frac{1}{\alpha}\int_{0}^\alpha VaR^s(X)ds,
\end{align*}
where $\alpha \in (0,1)$ is the significance level, and     ~$VaR^{\alpha}(X) = -F_X^{-1}(\alpha)$, i.e.,  the maximum expected loss for a given period and significance level. Its acceptance set is $\mathcal{A}_{VaR^\alpha}=\left\lbrace X\in L^1:\mathbb{P}(X<0)\leq\alpha\right\rbrace$. ES represents the expected value of a loss, given it is beyond the  $\alpha$-quantile of interest, i.e., $VaR$. 
We have \[\mathcal{A}_{ES^\alpha}=\left\lbrace X\in L^1:\int_0^\alpha VaR^s(X)ds\leq0\right\rbrace.\]
For ES, the dual set is $$\mathcal{Q}_{ES^\alpha}=\left\lbrace\mathbb{Q}\in\mathcal{Q} : \frac{d\mathbb{Q}}{d\mathbb{P}}\leq\frac{1}{\alpha} \right\rbrace.$$
\end{Exm}

\begin{Exm} (Expectile Value at Risk). \label{ex:EVaR} Expectile Value at Risk  (EVaR) links to the concept of an expectile, which is a generalization of the quantile function used for VaR estimation. 
EVaR is a functional $EVaR: L^1 \rightarrow \mathbb{R}$ directly defined as an argmin of a scoring function, is given by
	\begin{align*}
	EVaR^{\alpha}(X)&=-\argmin\limits_{x\in\mathbb{R}} E[\alpha[(X-x)^+]^2+(1-\alpha)[(X-x)^-]^2].
	\end{align*} 
In accordance with \cite{Bellini2014}, the EVaR is a law invariant coherent risk measure for $\alpha\leq0.5$. In addition, this measure is the only example of elicitable coherent risk measure beyond EL. 
\cite{Bellini2017} points out that according to EVaR, the position is acceptable when the ratio between the expected value of the gain and the loss is sufficiently high.
 In this case, we have \[\mathcal{A}_{EVaR^\alpha}=\left\lbrace X\in L^1\colon\frac{E[X^+]}{E[X^-]}\geq\frac{1-\alpha}{\alpha} \right\rbrace.\]The dual set of EVaR can be given by \[\mathcal{Q}_{EVaR^\alpha}=\left\lbrace\mathbb{Q}\in\mathcal{Q} \colon\:\exists\: a>0,\: a\leq\frac{d\mathbb{Q}}{d\mathbb{P}}\leq a\frac{1-\alpha}{\alpha} \right\rbrace.\] 
\end{Exm}

\begin{Exm}(Maximum Loss). Maximum Loss (ML) is the most extreme coherent risk measure. It is a functional $ML: L^\infty \rightarrow \mathbb{R}$ defined as
\begin{align*}
ML(X)=-\operatorname{ess}\inf X.
\end{align*}
ML leads to more protective situations since  $ML(Y)\geq\rho(Y)$ for any coherent risk measure $\rho$. For this measure, the dual set is given by
$\mathcal{Q}_{ML}=\mathcal{Q},$
i.e., all beliefs are considered. This measure does not directly fit into our framework since the supremum in its dual representation is not necessarily attained because it does not has finiteness assured in any $L^p,\:p\in[1,\infty)$. 
\end{Exm}

\subsection{Score functions}\label{sub:scfunc}

In this subsection, we present some examples of $S$. We describe possible but not limited choices for $S$. We also commented on some functions that do not fit our approach to avoid leaving out important scores.

\begin{Exm} [Squared Error]\label{ex:media}
	The squared error is one of the most common score functions. It is tied to the standard tools of least-squares regression. It is well known that its minimum is the variance, and its minimizer is the expectation. In our setup, we have that:
	\begin{align*}
	S_{EL}(x,y) &= (x-y)^2, \;  f_{S_{EL}} (x) = x^2, \\
	R_{E ,S_{EL}} (X) &=- \min \{ \argmin_{y \in \R}  \E[(X-y)^2] \} = \E[-X],  \\
	D_{E ,S_{EL}}(X) &= \min_{y \in \R} \E [(X-y)^2] = \sigma^2 (X).
	\end{align*}
	It is clear that $S_{EL}$ is a scoring function when $L^p \subseteq L^2$. Furthermore, $y \mapsto S_{EL}(x,y)$ is, for any $x \in \R$, differentiable with strictly increasing derivative and $\frac{\partial f_{S_{EL}}(x)}{\partial x}$ is convex. However, $S_{EL}$ is not positive homogeneous, while the resulting $R_{E ,S_{EL}}$ has this property, the deviation $D_{E ,S_{EL}}$ does not. This also yields that $R_{\rho, S_{EL}}$ and $D_{\rho, S_{EL}}$ are convex risk/deviation measures for any coherent $\rho$.
\end{Exm}

\begin{Exm} [Value at Risk] \label{ExampleVaR} Value at Risk  is a well known monetary risk measure in academia and industry. This measure is elicitable.
 As seen in \cref{ex:VaR}, VaR is defined as the lower quantile. 
Its connection to quantile regression is self-evident. 
	\begin{align}
	S_{VaR^\alpha}  (x,y) &= \alpha(x - y)^+ + (1 - \alpha)(x - y)^-, \; f_{S_{VaR^\alpha}} (x) = \alpha x^+ + (1-\alpha) x^-, \label{score1}
	\\
	R_{S_{VaR^\alpha}} (X) &=- \min \{ \argmin_{y \in \R}  \E[\alpha(X - y)^+ + (1 - \alpha)(X - y)^-] \} = VaR^\alpha (X),  \\
	D_{S_{VaR^\alpha}}(X) &= \min_{y \in \R} \{ \E[\alpha(X - y)^+ + (1 - \alpha)(X - y)^-] \}= ESD^\alpha(X),
	\end{align}
	where $ESD^\alpha$ is the expected shortfall deviation, a generalized deviation measure based on the expected shortfall, defined as:
	\[ ESD^\alpha (X) = ES^\alpha (X - \E[X]) . \]
	$S_{VaR^\alpha}$ is a scoring function when $L^p \subseteq L^1$.  Furthermore, $y \mapsto S_{VaR^\alpha}(x,y)$ is not, for any $x \in \R$ differentiable with strictly increasing derivative for all $y$, in particular, it is not differentiable at $y=x$. This violates the assumption in item (iii) of \cref{prp:prob}, which would yield that the set $B_X$ is a singleton; indeed, it is well known that quantiles are intervals. It  also prevents us from applying an item (ii) of  \Cref{prp:sol2} which would yield convexity, and again, it is well known that VaR is not convex. However, to guarantee that $D_{S_{VaR^\alpha}}$ is convex and positive homogeneous, it is enough for $f_{S_{VaR^\alpha}}$ to be convex and positive homogeneous, which is the case. For different choices of $\rho$, our results guarantee that $R_{\rho, S_{VaR^\alpha}}$ is a positive homogeneous monetary risk measure; and $D_{\rho, S_{VaR^\alpha}}$ is a generalized deviation measure.
	\begin{align*}
	R_{\rho.S_{VaR^\alpha}} (X) &=- \min \{ \argmin_{y \in \R}  \rho(-\alpha(X - y)^+ - (1 - \alpha)(X - y)^-) \},  \\
	D_{\rho,S_{VaR^\alpha}}(X) &= \min_{y \in \R} \{ \rho( -\alpha(X - y)^+ - (1 - \alpha)(X - y)^-) \}.
	\end{align*}
\end{Exm}

\begin{Exm}[Absolute Error]\label{ex:mediana}
	The Median is a special case of the VaR with $\alpha = 0.5$. In such a case, the VaR scoring function degenerates to the Absolute Error. The Median and the Absolute Error  are widely used in evaluating point forecasts. See \cite{gneiting2011making}. 
	\begin{align*}
	AE (x,y) &= | x-y|, \; f_{AE} (x) = |x|, \\
	R_{AE} (X) &= VaR^{\frac{1}{2}} (X) = \text{median} (X), \\
	D_{AE} (X)  &= ESD^{\frac{1}{2}} (X).
	\end{align*}
	Of course, the same characteristics and issues of the VaR scoring function are carried over the Absolute Error. Keeping the Absolute Error as the scoring function, we can change the risk measure inside the argmin  from the expected loss to the Maximum Loss, defined as $ML (X) = \esssup (-X)$. The ML is only well defined in $L^\infty$ and thus not directly fitted into our approach. Nonetheless, in this case,  we could define
	\begin{align*}
	R_{ML,AE} (X) &= - \min \{ \argmin_{y \in \R} \esssup | X-y|\} = \dfrac{1}{2}(\esssup X + \essinf X), \\
	D_{ML,AE} (X) &= - \min \{ \argmin_{y \in \R} \esssup | X-y|\} = \dfrac{1}{2}(\esssup X - \essinf X).
	\end{align*}
	We have that $R_{ML,AE}(X)$ is the center of the range of $X$. It is also worth noting that $R_{ML,AE}=R_{ML,S_{EL}} $. Additionally, $D_{ML,AE}(X)$ is the full range of $X$ and is a generalized deviation measure. 
\end{Exm}

\begin{Exm} [Minimum and Maximum Loss, a non-example]
	The Minimum Loss (MinL) is a very forgiven coherent risk measure, it is defined as $MinL(X) = \essinf -X = - \esssup X$,  it can be seen as a VaR with $\alpha = 1$, note that in this case, while the left quantile is well defined (for $X \in L^\infty$) the right quantile assumes $\infty$.
	\begin{align*}
	S_{MinL} (x,y) &= (x-y)^+, \; f_{S_{MinL}} (x) = x^+, \\
	R_{S_{MinL}} (X) &= MinL(X),
	\\
	D_{S_{MinL}} (X) &= 0.
	\end{align*}
	This example is a clear warning to the importance of using a scoring function and not any seemingly fine function, note that $D_{S_{MinL}}$ is identically $0$ that is because $S_{MinL}$ is not a scoring function, it does not fulfill requirement (i) of Definition 1 as $S_{MinL} (x,y) = 0 $  for all $y \geq x$. It is relevant to highlight that under $S_{MinL}$, $B_X = [\sup X, \infty)$. If we were to allow our scoring function to assume $\infty$ and not be continuous,  we could use the following:
	
	\begin{align*}
	S_{MinL} (x,y) &=  \begin{cases}  y-x \; &\text{ if }  \; x \leq y \\  \infty \; &\text{ if }  \; x > y
	\end{cases}, \\ 
	f_{S_{MinL}} (x) &=  \begin{cases} - x \; &\text{ if }  \; x \leq 0 \\  \infty \; &\text{ if }  \; x > 0
	\end{cases}, \\
	R_{S_{MinL}} (X) &= MinL(X),
	\\
	D_{S_{MinL}} (X) &= \sup X - \E[X].
	\end{align*}
		We can similarly obtain the Maximum Loss.
	\begin{align*}
	S_{ML} (x,y) &=  \begin{cases}  x-y \; &\text{ if }  \; x \geq y, \\  \infty \; &\text{ if }  \; x < y.
	\end{cases}, \\ 
	f_{S_{ML}} (x) &=  \begin{cases} x \; &\text{ if }  \; x \geq 0, \\  \infty \; &\text{ if }  \; x < 0.
	\end{cases}, \\
	R_{S_{ML}} (X) &= ML(X),
	\\
	D_{S_{ML}} (X) &= \E[X] - \inf X.
	\end{align*}
	$D_{S_{MinL}}$ and $D_{S_{ML}}$ are both generalized deviation measures, the former is known as upper range deviation and the latter as lower range deviation. It is worth recalling that both Minimum and Maximum Loss are only well defined risk measures for $L^\infty$.
\end{Exm} 

\begin{Exm} [Huber Loss]
	The Huber loss is a loss function used in robust regression. This is less sensitive to outliers in data than the squared error loss. In a sense, this interpolates between \cref{ex:media} and \ref{ex:mediana}. For a scaling parameter $\beta >0$, the $\beta$-truncation function $T_\beta$ and the Huber loss-like scoring function $HL$ are defined as:
	\begin{align*}
	T_\beta (x) &= \begin{cases}  \beta \; &\text{ if } \; x \geq \beta , \\  x \; &\text{ if } \; |x| \leq \beta ,
	\\ -\beta \; &\text{ if } \; x \leq \beta .
	\end{cases}, \\
	HB(x,y) &= \begin{cases} |x-y| - \dfrac{\beta}{2} \; &\text{ if } |x-y| \geq \beta, \\
	\dfrac{(x-y)^2}{2 \beta} \; &\text{ if } \; |x-y| \leq \beta.
	\end{cases}, \\
	f_{HB} (x) &= \begin{cases} |x| - \dfrac{\beta}{2} \; &\text{ if } |x| \geq \beta, \\
	\dfrac{x^2}{2 \beta} \; &\text{ if } \; |x| \leq \beta.
	\end{cases}, \\
	D_{HB} (X) &=  \E[f_{HB} (X) - R_{HB} (X) ],
	\end{align*}
	where $R_{HB} (X)$ is such that $ \E[T_\beta (X-R_{HB}(X))] =0$. This relationship was highlighted in \cite{Rockafellar2013}. The $f_{HB}$ scoring function is quadratic for small values of $x$ and linear for large values. The Huber Loss function  combines much of the sensitivity of the mean-unbiased, minimum-variance estimator of the mean (using the quadratic loss function) and the robustness of the median-unbiased estimator (using the absolute value function). As $HB$ is a  scoring function such that $y \mapsto HB(x,y)$ is, for any $x \in \R$, differentiable with strictly increasing derivative and $\frac{\partial f_{HB}(x)}{\partial x}$ is convex, $R_{\rho,HB}$ is a convex risk measure and $D_{\rho,HB}$ is a convex deviation measure for any coherent risk measure $\rho$.
	
\end{Exm}

\begin{Exm}[LINEX] 
	The Entropic risk measure (ENT) is a risk measure that depends on the risk aversion of the user through the exponential utility function. It is a prime example of a convex risk measure that is not coherent. This measure is the map $ENT^\gamma\colon L^1\to\mathbb{R}$  defined as \[ENT^\gamma (X) = \frac{1}{\gamma} \log \E [ e^{-\gamma X}]\] for a risk aversion parameter $\gamma >0$. It is associated with the linear-exponential loss function (LINEX). The intuition behind this loss is that it is an asymmetric approximation to the usual quadratic loss function. This is a popular loss function in econometrics.
	\begin{align*}
	S_\gamma (x,y) &=  e^{\gamma (y-x) } - \gamma (y-x) -1  , \\
	f_{S_\gamma} (x) &=  e^{-\gamma x} +\gamma x -1, \\
	R_{S_\gamma} (X) &= ENT^\gamma (X) , \\
	D_{S_\gamma} (X) &= ENT^\gamma(X - \E[X]).
	\end{align*}
	$S_\gamma$ is a  scoring function such that $y \mapsto S_\gamma(x,y)$ is, for any $x \in \R$, differentiable with strictly increasing derivative and $\frac{\partial f_{S_\gamma}(x)}{\partial x}$ is convex. Hence, for any coherent $\rho$, $R_{\rho,S_\gamma}$ is a convex risk measure and $D_{\rho,S_\gamma}$ is a convex deviation measure. For more details on LINEX score functions, see \cite{zellner1986bayesian}.
\end{Exm}

\begin{Exm} [Expectile and Variantile] 
Expectile Value at Risk (see \cref{ex:EVaR})  and Variantile are directly defined as an argmin and a minimum, respectively, for a given scoring function.
EVaR has a score function analog  quadratic form of the VaR score function (see \cref{ExampleVaR}).  The EVaR arises as a solution for asymmetric least squares. Taking $\alpha = 0.5$, we recover traditional least squares, and the EVaR coincides with the Expected Loss (see \cref{ex:EL}).
	\begin{align*}
	S_{EVaR^\alpha} (x,y) &= \alpha [(x-y)^+]^2 + (1-\alpha) [(x-y)^-]^2, \\
	f_{S_{EVaR^\alpha}} (x)  &= \alpha [(x)^+]^2 + (1-\alpha) [(x)^-]^2.
	\end{align*}
	The EVaR is defined as $R_{S_{EVaR^\alpha}}$ and the Variantile is defined as $D_{S_{EVaR^\alpha}}$; the first is a coherent risk measure while the second is a generalized deviation measure
\end{Exm}

\begin{Exm}
	We now present a score which is  a generalization of the Cauchy/Lorentzian,
	Geman-McClure, Welsch/Leclerc, generalized Charbonnier, Charbonnier/pseudo-Huber/L1-L2, and L2 loss functions. This score function was proposed by \cite{barron2019}, and it is defined as 
	\begin{align*}
	f_{S_\alpha} (x) &= \begin{cases} \dfrac{x^2}{2}  &\text{ if } \alpha =2,  \\
	\log \left( \dfrac{x^2}{2} +1 \right) &\text{ if } \alpha = 0 , \\
	1- \exp \left( - \dfrac{x^2}{2} \right) &\text{ if } \alpha = - \infty, \\
	\dfrac{| \alpha - 2| }{\alpha} \left( \left( \dfrac{x^2}{|\alpha - 2 | } +1 \right)^{\frac{\alpha}{2}} -1 \right) &\text{ otherwise. }
	\end{cases} \\
	S_\alpha (x,y) &= f_{S_\alpha}(x-y).
	\end{align*}
	The map $\alpha \mapsto S_\alpha (x,y)$ is continuous in the parameter $\alpha$, which is a shape parameter that controls the robustness of the loss. When $\alpha = 2$ this loss resembles the $S_{EL}$ loss, when $\alpha = 1$ it is a a pseudo Huber loss, see \cite{huber1992}. $\alpha = 0$ yields the Cauchy or Lorentzian loss. $S_{-2}$ is the Geman-McClure loss. Lastly, letting $\alpha = - \infty$ yields the Welch loss (see \cite{dennis1978}). It is convex, hence, a scoring function in our framework for $\alpha \geq 1$. We highlight that $S_\alpha$ is a symmetric loss for all $\alpha$. 
\end{Exm} 

\begin{Exm} [Absolute percentage error and relative error]
	Absolute percentage error (APE) and relative error (RE) are also widely used score functions to assess point forecasts. Both scores are defined in the following way
	\begin{align*}
	APE(x,y) &=   \dfrac{|x-y|}{|y|},\\
	RE(x,y) &=   \dfrac{|x-y|}{|x|}.
	\end{align*} 
	However, neither fits our approach as there is no real function $f \colon \R \rightarrow \R$ such that $APE(x,y) = f(x-y)$ or $RE(x,y)=f(x-y)$.
\end{Exm}

\begin{Exm} [Location of minimum variance squared distance]
	\cite{landsman2022} proposed the Location of minimum variance squared distance (LVS) to measure multivariate risk. LVS is defined as 
	\begin{align*}
	LVS (\textbf{X}) = \argmin_{y \in \R^n} Var (|| \textbf{X}-y||^2),
	\end{align*}
	when $\textbf{X}$ is a vector of $n$ random variables, $|| . ||$ denotes the Euclidean distance, and $Var$ is the variance. For $n=1$,  LVS becomes similar to our approach, i.e., 
	\begin{align*}
	LVS (X) = \argmin_{y \in \R} Var ( (X-y)^2 ) = \argmin_{y \in \R}  { \sigma( (X-y)^2 )},
	\end{align*}
	where $\sigma$ is the standard deviation. This is in a similar spirit to our approach. However, it uses a generalized deviation measure (variance) instead of a coherent risk measure. Therefore, it cannot be used in our approach.
\end{Exm}

\begin{Exm}[Co-elicitability]
	Unfortunately, the ES (see \cref{ex:VaR}) is not elicitable, that is, there is no score function such that the ES is its minimizer. However, as shown by \cite{fissler2016higher} the function $T (X) \mapsto (VaR^\alpha(X), ES^\alpha (X) ) \in \R^2$ is. In this case, the score function has its domain in $\R^3$.
	Considering the family of scoring functions for ES proposed by \cite{fissler2016higher}, \cite{gerlach2017semi} suggest using the following score function to assess ES point forecasts
	\begin{align*}
	S_{ES^\alpha} (x,y,z) &= y(1_{x<y}  - \alpha) - x1_{x<y}  + e^z  \left( z-y + \dfrac{1_{x<y} }{\alpha} (y-x)  \right) - e^z + 1 - \log(1-\alpha),
	\\ (VaR^\alpha (X) , ES^\alpha(X) ) &= \argmin_{ y \in \R^2} \E[ S_{ES^\alpha} (X,y)].
	\end{align*}	
	On the same topic, the Range Value at Risk (RVaR) proposed by \cite{Cont2010}, is defined as  \[ RVaR^{\alpha , \beta} (X) = -\frac{1}{\beta-\alpha}\int_\alpha^\beta F_X^{-1}(s)ds,\:0\leq\alpha\leq\beta\leq 1.\]
Similar to ES, this measure is not elicitable but, $T(X) \mapsto (VaR^\alpha, VaR^\beta , RVaR^{\alpha,\beta} )$ is, as shown by \cite{fissler2021}, under the following score
	\begin{align*}
	&S_{RVaR^{\alpha,\beta}} (x,y,z,w) \\
	 = & y ( 1_{x<y} -\alpha ) -x1_{x<y} + z(1_{x<z} - \beta) - x1_{x<z} -\log(\cosh((\alpha-\beta)w)) - \log(1-\alpha) \\ + &(\beta-\alpha) \tanh((\beta-\alpha)w) \left[ w+\dfrac{1}{\beta-\alpha} \left( S_{VaR^\beta} (x,z) - S_{VaR^\alpha} (x,y) \right) \right] +1.
	\end{align*}
	Clearly, the scores above do not fit our framework, as we limit ourselves to scores with domain in $\R^2$. Nonetheless, it is possible to extend the framework for the domain of $S$ be $\R^{k+1}$, where $k$ is the degree of co-elicitability in the sense of \cite{fissler2016higher}, with the minimization taken on $\mathbb{R}^k$. In this setup, one can define $R$ as the coordinate of interest from the argmin vector. For instance, $ES$ implies $k=2$ and the second coordinate, while for $RVaR$ one gets $k=3$.
	
\end{Exm}

\begin{Exm}[Cost minimization]
\cite{Righi2020} brings forward a similar approach. They propose a robust risk measurement approach that minimizes the expectation of overestimation ($G$) and underestimation ($L$) costs. Their loss function intrinsically depends on the exogenous random variables $G$ and $L$. The study also proposes a deviation measure that is strikingly similar to ours. In fact, if $G$ and $L$ were to be taken as constants, it is equal to $D_{\rho,S_{VaR^\alpha}}$; under some mild conditions and proper choice of dual set, see their Proposition 3. For suitable $G$, $L$ and dual set $\mathcal{Q}^\prime$, their score function, positive homogeneous monetary risk measure, and generalized deviation measure are defined as
\begin{align*}
S_{GL} (x,y) &= (x-y)^+G + (x-y)^-L, \\
GL(X)&=   \sup\limits_{\mathbb{Q}\in\mathcal{Q}^\prime}\left\lbrace -\min\left\lbrace\arg\min\limits_{x\in\mathbb{R}} E_{\mathbb{Q}}[(X-x)^{+}G+(X-x)^{-}L]\right\rbrace\right\rbrace ,
\\
GLD(X)&=\sup\limits_{\mathbb{Q}\in\mathcal{Q}^\prime}\left\lbrace \min\limits_{x\in\mathbb{R}} E_{\mathbb{Q}}[(X-x)^{+}G+(X-x)^{-}L]\right\rbrace =  \min\limits_{x\in\mathbb{R}} \rho((X-x)^{+}G+(X-x)^{-}L).
\end{align*}
$S_{GL}$ is a scoring function in our framework  if $G \in (0,1)$ and $L = 1-G$. In this case, $S_{GL} = S_{VaR^G}$.
\end{Exm}

\section{Conditional risk, linear regression and optimal portfolio weights}\label{CondRisk}

In this section, we define the conditional version of risk and deviation measures, which can be obtained from the connection between regression models and our framework. 
We discuss the properties that conditional risk respects.
Moreover, we formalize a connection between minimum deviation portfolio optimization and regression model considering the score that results from the deviation measure. We also provide a solution for the optimal replication hedging problem based on conditional risk measures obtained by the connection between our approach and linear regression. 

 We now provide a link between linear regression and our framework. 

\begin{Def}\label{def:eli2}
	Let $X\in  (L^p)^n$ and condition (iii) in  \Cref{prp:prob} holds. The conditional risk is a map $R_{\rho,S}\colon L^p\rightarrow L^p$ given by
	\begin{align}
	R_{\rho,S}(Y|X)&=-\left(\mu^{*}+\sum_{i=1}^{n}\beta^{*}_iX_i\right),\nonumber\\
	\text{where}&\;(\mu^{*},\beta_1^{*},\dots,\beta_n^{*})= \argmin\limits_{\mu\in\mathbb{R},\beta\in\mathbb{R}^n}\rho\left(-S\left(Y,\mu+\sum_{i=1}^{n}\beta_iX_i\right) \right).
	\label{eq:eli3}
	\end{align}
\end{Def}

\begin{Rmk}\label{rmk:opt}
	In the case  $\mathcal{Q}=\{\mathbb{P}\}$  and $S=S_{VaR^\alpha}$ or $S=S_{EVaR^\alpha}$, we obtain, respectively, the conditional $\alpha$ quantile and $\alpha$ expectile obtained from quantile regression and expectlie regression. In particular, with $S=S_{EVaR^{0.5}}=S_{EL}$, one recovers the usual ordinary least squares linear regression. We  also define a conditional version of the deviation $D_{\rho,S}\colon L^p\rightarrow L^p$ as
	\[
	D_{\rho,S}(Y|X)=\min\limits_{\mu\in\mathbb{R},\beta\in\mathbb{R}^n}\rho\left(-S\left(Y,\mu+\sum_{i=1}^{n}\beta_iX_i\right) \right)=\rho\left(-S\left(Y,-R_{\rho,S}(Y|X)\right) \right).\] Since it is not directly used in the context we consider, we do not fully explore it due to parsimony.
\end{Rmk}

We now explore properties of the conditional risk similarly as in Propositions \ref{prp:prob} and \ref{prp:sol2}.

\begin{Prp}\label{prp:reg}
	Let $\beta(Y,X)$ be the argmin in \cref{eq:eli3} for $Y\in L^p$, $X=(X_1,\dots,X_n)\in (L^p)^n$. We have the following:
	\begin{enumerate}
		\item  $R_{\rho,S}(Y|X)$ is well defined. 
		\item $\mu^{*}=-R\left(Y-\sum_{i =1}^n\beta^{*}_iX_i\right)$.
		\item if  $Y\leq Z$, then $R(Y|X)\geq R(Z|X)$ for any $Y,Z\in L^p$, and $R(Y+C|X)=R(Y|X)-C$ for any $C\in\mathbb{R}$.
		\item $\beta( Y+CX,X)=\beta(Y,X)+C$ for any $C\in\mathbb{R}^n$. Hence, $R_{\rho,S}(Y+CX|X)=R_{\rho,S}(Y|X)-CX$.
		\item  if $\frac{\partial f_S(x)}{\partial x}$ is convex, then $R(\lambda Y+(1-\lambda)Z|X)\leq \lambda R(Y|X)+(1-\lambda)R(Z|X)$ for any $Y,Z\in L^p$ and any $\lambda \in[0,1]$. 
		\item if $f_S$ is positive homogeneous, then $\beta(\lambda Y,X)=\lambda\beta(Y,X)$ for any $\lambda\geq 0$. Hence, $R_{\rho,S}(\lambda Y|X)=\lambda R_{\rho,S}(Y|X)$.
		\item $\beta( Y,XA)=A^{-1}\beta(Y,X)$ for any $n\times n$ non-singular matrix $A$.	
		\item $R_{\rho,S}(Y|X)=-\left(\mu^{*}+\sum_{i=1}^{n}\beta^{*}_iX_i\right)$ if and only if $Y=\mu^{*}+\sum_{i=1}^{n}\beta^{*}_iX_i+\epsilon$, where $R_{\rho,S}(\epsilon|X)=0$.
		
	\end{enumerate}
\end{Prp}

\begin{proof}
	For (i), we have that $-\left(\mu+\sum_{i=1}^{n}\beta_iX_i\right)\in L^p$ for any $(\mu,\beta_1,\dots,\beta_n)$. Further, defining $f\colon\mathbb{R}^{n+1}\to\mathbb{R}$ as $f(\mu,\beta_1,\dots,\beta_n)=\rho\left(-S\left(Y,\mu+\sum_{i=1}^{n}\beta_iX_i\right) \right)$, the deduction to prove that the argmin is a singleton is similar to the one in  \Cref{prp:prob}, but adapted to $\mathbb{R}^{n+1}$.
	
	Regarding (ii), note that for any $\beta\in\mathbb{R}^{n}$ we have that \[\argmin\limits_{\mu\in\mathbb{R}}\rho\left(-f_S\left(Y-\sum_{i=1}^{n}\beta_i(X_i)-\mu\right) \right)=-R_{\rho,S}\left(Y- \sum_{i=1}^{n}\beta_i(X_i)\right).\] Thus, we obtain that
	\begin{align*}
	&\argmin\limits_{\mu\in\mathbb{R},\beta\in\mathbb{R}^{n}}\rho\left(-f_S\left(Y-\sum_{i=1}^{n}\beta_i(X_i)-\mu\right) \right)\\
	=&\argmin\limits_{\beta\in\mathbb{R}^{n}}\rho\left(-f_S\left(Y-\sum_{i=1}^{n}\beta_i(X_i)+R_{\rho,S}\left(Y- \sum_{i=1}^{n}\beta_i(X_i)\right)\right) \right)\\
	=&\rho\left(-f_S\left(Y-\sum_{i=1}^{n}\beta^{*}_i(X_i)+R_{\rho,S}\left(Y- \sum_{i=1}^{n}\beta^{*}_i(X_i)\right)\right) \right). 
	\end{align*}
	Hence, $\mu^{*}=-R\left(Y-\sum_{i =1}^n\beta^{*}_iX_i\right)$.
	
	Concerning (iii), let $Y\leq Z$. By   \Cref{prp:sol2}, we have $\mathbb{P}-a.s.$ that \[R(Y|X)(\omega)=R\left(Y-\mu^{*}-\sum_{i=1}^n\beta^{*}_iX_i(\omega)\right)\geq R\left(Z-\mu^{*}-\sum_{i=1}^n\beta^{*}_iX_i(\omega)\right)= R(Z|X)(\omega).\] Hence,  $R(Y|X)\geq R(Z|X)$. Further, for any $C\in\mathbb{R}$ we have that \begin{align*}
	&\argmin\limits_{\mu\in\mathbb{R},\beta\in\mathbb{R}^n}\rho\left(-f_S\left( Y+C-\mu-\sum_{i=1}^{n}\beta_iX_i\right)\right)\\
	=&\argmin\limits_{\mu-C\in\mathbb{R},\beta\in\mathbb{R}^n}\rho\left(-f_S\left( Y-\mu+\sum_{i=1}^{n}\beta_iX_i\right) \right)\\
	=&\argmin\limits_{\mu\in\mathbb{R},\beta\in\mathbb{R}^n} \rho\left(-f_S\left( Y-\mu+\sum_{i=1}^{n}\beta_iX_i\right) \right)+(C,0).
	\end{align*}
	Then,  $R(Y+C|X)=R(Y|X)-C$.

	Regarding (iv), for any $C\in\mathbb{R}^n$, we have that 
	\begin{align*}
	&\argmin\limits_{\mu\in\mathbb{R},\beta\in\mathbb{R}^n}\rho\left(-f_S\left( Y+\sum_{i =1}^nC_iX_i-\mu-\sum_{i=1}^{n}\beta_iX_i\right)\right)\\
	=&\argmin\limits_{\mu\in\mathbb{R},\beta\in\mathbb{R}^n}\rho\left(-f_S\left( Y-\mu+\sum_{i=1}^{n}(\beta_i-C_i)X_i\right) \right)\\
	=&\argmin\limits_{\mu\in\mathbb{R},\beta\in\mathbb{R}^n} \rho\left(-f_S\left( Y-\mu+\sum_{i=1}^{n}\beta_iX_i\right) \right)+(0,C)
	\end{align*}
	Thus, $R_{\rho,S}(Y+CX|X)= -\left( \mu^{*}+\sum_{i=1}^{n}(\beta^{*}_i-C_i)X_i\right) =R_{\rho,S}(Y|X)-CX$.
	
	Concerning (v), the claim follows similarly to that in  \Cref{prp:sol2} by considering the f.o.c. \[\dfrac{\partial \rho(-f_S(Y-\sum_{i =1}^n\beta_iX_i-R(Y-\sum_{i =1}^n\beta_iX_i)))}{\partial \beta_i}=0,\:i=1,\dots,n.\]  Let then $g\colon L^p\times (L^p)^n\to\mathbb{R}^n$ be as 
	\[g(Y,X)=\dfrac{\partial \rho(-f_S(Y-\sum_{i =1}^n\beta_iX_i-R(Y-\sum_{i =1}^n\beta_iX_i)))}{\partial \beta_i}(\beta),\:i=1,\dots,n\] which is convex in its domain and non-increasing in $X$ for any $Y\in L^p$. Let $\lambda\in[0,1]$ and $Y,Z\in L^p$. Then we have 
	\begin{align*}
	&g(\lambda Y+(1-\lambda)Z,\lambda R(Y|X)+(1-\lambda)R(Z|X))\\
	\leq& \lambda g(Y,R(Y|X))+(1-\lambda)g(Z,R(Z|X))=0\\
	=&g(\lambda Y+(1-\lambda)Z,R(\lambda Y+(1-\lambda)Z|X)).
	\end{align*} Thus, we obtain $R(\lambda Y+(1-\lambda)Z|X)\leq\lambda R(Y|X)+(1-\lambda)R(Z|X)$.
	
	For (vi), if $\lambda =0$ the result is trivial. Further, we have for any $\lambda> 0$ that \begin{align*}
	&\argmin\limits_{\mu\in\mathbb{R},\beta\in\mathbb{R}^n}\rho\left(-f_S\left(\lambda Y-\mu-\sum_{i=1}^{n}\beta_iX_i\right)\right)\\
	=&\argmin\limits_{\mu\in\mathbb{R},\beta\in\mathbb{R}^n}\lambda \rho\left(-f_S\left( Y-\dfrac{\mu+\sum_{i=1}^{n}\beta_iX_i}{\lambda}\right) \right)\\
	=&\lambda\argmin\limits_{\mu\in\mathbb{R},\beta\in\mathbb{R}^n}\rho\left(-f_S\left( Y-\mu-\sum_{i=1}^{n}\beta_iX_i\right)\right).
	\end{align*}We thus get that
	\[R(\lambda Y|X)=\lambda\left(-\mu^{*}-\sum_{i=1}^n\beta^{*}_iX_i\right)=\lambda R(Y|X).\]
	
	Concerning (vii), we have  for any $n\times n$ non-singular matrix $A$ that 	
	\begin{align*}
	&\argmin\limits_{\mu\in\mathbb{R},\beta\in\mathbb{R}^n}\rho\left(-f_S\left( Y-\mu-\beta X A\right)\right)\\
	=&\argmin\limits_{\mu\in\mathbb{R},\beta\in\mathbb{R}^n} \rho\left(-f_S\left( Y-\mu-(\beta A )X\right) \right)\\
	=&(1,A^{-1})\argmin\limits_{\mu\in\mathbb{R},\beta\in\mathbb{R}^n} \rho\left(-f_S\left( Y-\mu+\sum_{i=1}^{n}\beta_iX_i\right) \right).
	\end{align*}
	
	For (viii), if $Y=\mu^{*}+\sum_{i=1}^{n}\beta^{*}_iX_i+\epsilon$ with $R_{\rho,S}(\epsilon|X)=0$, then it is direct that \[0=R_{\rho,S}(\epsilon|X)=R_{\rho,S}\left(Y-\left(\mu^{*}+\sum_{i=1}^{n}\beta^{*}_iX_i\right)\bigg|X\right)=R_{\rho,S}(Y|X)+\mu^{*}+\sum_{i=1}^{n}\beta^{*}_iX_i.\] Thus, the if part of the claim follows. For the converse, let $\epsilon=Y-\left(\mu^{*}+\sum_{i=1}^{n}\beta^{*}_iX_i\right)$. Then $Y=\mu^{*}+\sum_{i=1}^{n}\beta^{*}_iX_i+\epsilon$ and \[R_{\rho,S}(\epsilon|X)=R_{\rho,S}\left(Y-\left(\mu^{*}+\sum_{i=1}^{n}\beta^{*}_iX_i\right)\bigg|X\right)=R_{\rho,S}(Y|X)+\mu^{*}+\sum_{i=1}^{n}\beta^{*}_iX_i=0.\]
\end{proof}

\begin{Rmk}
Let $\sigma(X)\subseteq\mathcal{F}$ be the sub-sigma-algebra generated by $X$. It is straightforward to verify that the previous Proposition \ref{prp:reg} implies that $R(Y+C|X)=R(Y|X)-C$ for any $C\in L^p(\Omega,\sigma(X),\mathbb{P}])$. Furthermore, if if $f_S$ is positive homogeneous, then $R_{\rho,S}(\lambda Y|X)=\lambda R_{\rho,S}(Y|X)$ or any $\lambda\in L^p_+(\Omega,\sigma(X),\mathbb{P}])$. Thus, we indeed have that $R_{\rho,S}$ is in fact a conditional risk measure in the sense of \cite{Ruszczynski2006}. 
\end{Rmk}

\begin{Rmk}
Based on such framework, one can have metrics in our setup that are similar to the usual coefficient of determination $R^2$ as \[CD_{\rho,S}(Y,X)=1- \dfrac{\rho(-S(Y,-R(Y|X))}{\rho(-S(Y,-R(Y)))},\] where  $Y\in L^p$, $X=(X_1,\dots,X_n)\in (L^p)^n$. Such quantity can be used to summarize the association of $Y$ and $X$. Furthermore, it is also possible to study inference properties of estimated parameters $\beta(Y,X)$ as well as hypothesis tests such as counterparts to the usual $t$ and $F$ tests for OLS approaches. Such topics are outside our current scope and left for future research.
\end{Rmk}

We now formalize the minimum deviation and replication hedging problems to our framework and state a result for our setup that guarantees the existence of a solution and how to obtain it.

\begin{Def}\label{prp:Dev}
	Let $X=(X_1,\dots,X_n)\in (L^p)^n$. The minimum deviation portfolio optimization problem for $X$, $P(X)$, is defined as
	\begin{equation}\label{prp:dev}
	\min_{\substack{w\in\mathbb{R}^n
			\\\sum_{i =1}^nw_i=1}} D_{\rho,S}\left(\sum_{i =1}^nw_iX_i\right)
	\end{equation}
\end{Def}

\begin{Prp}
	We have $w^{*}=(w_1^*,\dots,w_n^{*})\in\argmin P(X)$ if and only if   $R_{\rho,S}(Y|(Y-X_1,\dots,Y-X_{n}))=-\left( \mu^{*}+\sum_{i=1}^{n}w^\prime_i(Y-X_i)\right) $, where $Y=\frac{1}{n}\sum_{i =1}^nX_i$, $w_i^{*}=w_i^\prime+\frac{1}{n}\left( 1-\sum_{i =1}^nw_i^\prime\right)$  and $\mu^{*}=R_{\rho,S}(\sum_{i=1}^nw^{\prime}_i(Y-X_i))$.
\end{Prp}
\begin{proof}
	We have by Definition that  \begin{align*}
	&R_{\rho,S}(Y|(Y-X_1,\dots,Y-X_{n}))=-\left( \mu^{*}+\sum_{i=1}^{n}w^{*}_i(Y-X_i)\right) \\
	\iff& (\mu^{*},w^{*})=\argmin\limits_{\mu\in\mathbb{R},w\in\mathbb{R}^{n}}\rho\left(-S\left(Y,\mu+\sum_{i=1}^{n}w_i(Y-X_i)\right) \right).
	\end{align*} Note that $w^\prime\in\mathbb{R}^n\iff w_i^{*}=w_i^\prime+\frac{1}{n}\left( 1-\sum_{i =1}^nw_i^\prime\right)\in\mathbb{R}\:\text{and}\:\sum_{i =1}^nw_i^{*}=1$. The equivalence then follows by:
	\begin{align*}
	&\min\limits_{\mu\in\mathbb{R},w\in\mathbb{R}^{n}}\rho\left(-S\left(Y,\mu+\sum_{i=1}^{n}w_i(Y-X_i)\right) \right)\\
	=&\min\limits_{\mu\in\mathbb{R},w\in\mathbb{R}^{n}}\rho\left(-f_S\left(\left(1-\sum_{i=1}^{n}w_i\right)Y+\sum_{i=1}^{n}w_iX_i-\mu\right)\right)\\
	=&\min\limits_{w\in\mathbb{R}^{n}}\rho\left(-f_S\left(\sum_{i=1}^{n}\left( \frac{1}{n}\left(1-\sum_{i=1}^{n}w_i\right)+w_i\right) X_i+R_{\rho,S}\left(\sum_{i=1}^{n}\left( \frac{1}{n}\left(1-\sum_{i=1}^{n}w_i\right)+w_i\right) X_i\right)\right) \right)\\
	=&\min\limits_{\substack{w\in\mathbb{R}^n
			\\\sum_{i =1}^nw_i=1}}\rho\left(-f_S\left(\sum_{i=1}^{n}w_iX_i-R_{\rho,S}\left(\sum_{i=1}^{n}w_iX_i\right)\right)\right)\\
	=&\min_{\substack{w\in\mathbb{R}^n
			\\\sum_{i =1}^nw_i=1}} D_{\rho,S}\left(\sum_{i =1}^nw_iX_i\right).
	\end{align*}
\end{proof}

	


\begin{Def}
	Let $Y\in L^p$ be given and $X=(X_1,\dots,X_n)\in (L^p)^n$. The optimal replication hedging problem for $X$, $H(X)$, is defined as
	\begin{equation}
	\min_{w\in\mathbb{R}^n}\rho\left(-S\left(Y,\mu+\sum_{i=1}^{n}w_iX_i\right) \right)
	\end{equation}
\end{Def}


\begin{Prp}\label{prp:hed}
	$w^{*}=(w_1^*,\dots,w_n^{*})\in\argmin H(X)$ if and only if   $R(Y|X)=-\left( \mu^{*}+\sum_{i=1}^{n}w_i^{*}X_i\right) $, where $\mu^{*}=-R(Y-\sum_{i=1}^nw^{*}_iX_i)$.
\end{Prp}
\begin{proof}
	We have by Definition that  \begin{align*}
	&R_{\rho,S}(Y|X_1,\dots,X_{n})=-\left( \mu^{*}+\sum_{i=1}^{n}w^{*}_iX_i\right) \\
	\iff& (\mu^{*},w^{*})=\argmin\limits_{\mu\in\mathbb{R},w\in\mathbb{R}^{n}}\rho\left(-S\left(Y,\mu+\sum_{i=1}^{n}w_iX_i\right) \right).
	\end{align*} Further, notice that \[\argmin\limits_{w\in\mathbb{R}^{n}}\rho\left(-S\left(Y,\sum_{i=1}^{n}w_iX_i\right) \right)=\argmin\limits_{w\in\mathbb{R}^{n}}\rho\left(-S\left(Y,k+\sum_{i=1}^{n}w_iX_i\right) \right),\:\forall\:k\in\mathbb{R}.\] We also have that 
	\[\min\limits_{\mu\in\mathbb{R},w\in\mathbb{R}^{n}}\rho\left(-S\left(Y,\sum_{i=1}^{n}w_iX_i\right) \right)
	=\min\limits_{w\in\mathbb{R}^n}\rho\left(-f_S\left(Y-\sum_{i=1}^{n}w_iX_i+R\left(Y-\sum_{i=1}^{n}w_iX_i\right)\right)\right).\]
	From these facts, we get the equivalence between both	$w^{*}=(w_1^*,\dots,w_n^{*})\in\argmin H(X)$ and   $R(Y|X)=-\left( \mu^{*}+\sum_{i=1}^{n}w_i^{*}X_i\right) $.
\end{proof}

	\bibliographystyle{elsarticle-harv}
	\bibliography{Scorereference}

\begin{thebibliography}{71}
\expandafter\ifx\csname natexlab\endcsname\relax\def\natexlab#1{#1}\fi
\providecommand{\url}[1]{\texttt{#1}}
\providecommand{\href}[2]{#2}
\providecommand{\path}[1]{#1}
\providecommand{\DOIprefix}{doi:}
\providecommand{\ArXivprefix}{arXiv:}
\providecommand{\URLprefix}{URL: }
\providecommand{\Pubmedprefix}{pmid:}
\providecommand{\doi}[1]{\href{http://dx.doi.org/#1}{\path{#1}}}
\providecommand{\Pubmed}[1]{\href{pmid:#1}{\path{#1}}}
\providecommand{\bibinfo}[2]{#2}
\ifx\xfnm\relax \def\xfnm[#1]{\unskip,\space#1}\fi
\bibitem[{Acerbi(2002)}]{Acerbi2002}
\bibinfo{author}{Acerbi, C.}, \bibinfo{year}{2002}.
\newblock \bibinfo{title}{Spectral measures of risk: A coherent representation
  of subjective risk aversion}.
\newblock \bibinfo{journal}{Journal of Banking \& Finance}
  \bibinfo{volume}{26}, \bibinfo{pages}{1505 -- 1518}.
\bibitem[{Artzner et~al.(1999)Artzner, Delbaen, Eber and Heath}]{Artzner1999}
\bibinfo{author}{Artzner, P.}, \bibinfo{author}{Delbaen, F.},
  \bibinfo{author}{Eber, J.M.}, \bibinfo{author}{Heath, D.},
  \bibinfo{year}{1999}.
\newblock \bibinfo{title}{Coherent measures of risk}.
\newblock \bibinfo{journal}{Mathematical Finance} \bibinfo{volume}{9},
  \bibinfo{pages}{203--228}.
\bibitem[{Balter and Pelsser(2020)}]{balter2020pricing}
\bibinfo{author}{Balter, A.G.}, \bibinfo{author}{Pelsser, A.},
  \bibinfo{year}{2020}.
\newblock \bibinfo{title}{Pricing and hedging in incomplete markets with model
  uncertainty}.
\newblock \bibinfo{journal}{European Journal of Operational Research}
  \bibinfo{volume}{282}, \bibinfo{pages}{911--925}.
\bibitem[{Barigou et~al.(2022)Barigou, Bignozzi and
  Tsanakas}]{barigou2022insurance}
\bibinfo{author}{Barigou, K.}, \bibinfo{author}{Bignozzi, V.},
  \bibinfo{author}{Tsanakas, A.}, \bibinfo{year}{2022}.
\newblock \bibinfo{title}{Insurance valuation: A two-step generalised
  regression approach}.
\newblock \bibinfo{journal}{ASTIN Bulletin: The Journal of the IAA}
  \bibinfo{volume}{52}, \bibinfo{pages}{211--245}.
\bibitem[{Barron(2019)}]{barron2019}
\bibinfo{author}{Barron, J.T.}, \bibinfo{year}{2019}.
\newblock \bibinfo{title}{A general and adaptive robust loss function}, in:
  \bibinfo{booktitle}{Proceedings of the IEEE/CVF Conference on Computer Vision
  and Pattern Recognition}, pp. \bibinfo{pages}{4331--4339}.
\bibitem[{Bellini and Bignozzi(2015)}]{Bellini2015}
\bibinfo{author}{Bellini, F.}, \bibinfo{author}{Bignozzi, V.},
  \bibinfo{year}{2015}.
\newblock \bibinfo{title}{On elicitable risk measures}.
\newblock \bibinfo{journal}{Quantitative Finance} \bibinfo{volume}{15},
  \bibinfo{pages}{725--733}.
\bibitem[{Bellini and {Di Bernardino}(2017)}]{Bellini2017}
\bibinfo{author}{Bellini, F.}, \bibinfo{author}{{Di Bernardino}, E.},
  \bibinfo{year}{2017}.
\newblock \bibinfo{title}{Risk management with expectiles}.
\newblock \bibinfo{journal}{The European Journal of Finance}
  \bibinfo{volume}{23}, \bibinfo{pages}{487--506}.
\bibitem[{Bellini et~al.(2014)Bellini, Klar, M\"uller and Gianin}]{Bellini2014}
\bibinfo{author}{Bellini, F.}, \bibinfo{author}{Klar, B.},
  \bibinfo{author}{M\"uller, A.}, \bibinfo{author}{Gianin, E.R.},
  \bibinfo{year}{2014}.
\newblock \bibinfo{title}{Generalized quantiles as risk measures}.
\newblock \bibinfo{journal}{Insurance: Mathematics and Economics}
  \bibinfo{volume}{54}, \bibinfo{pages}{41 -- 48}.
\bibitem[{Bellini et~al.(2021)Bellini, Laeven and Gianin}]{Bellini2019}
\bibinfo{author}{Bellini, F.}, \bibinfo{author}{Laeven, R.J.A.},
  \bibinfo{author}{Gianin, E.R.}, \bibinfo{year}{2021}.
\newblock \bibinfo{title}{Dynamic robust {Orlicz} premia and
  {Haezendonck–Goovaerts} risk measures}.
\newblock \bibinfo{journal}{European Journal of Operational Research}
  \bibinfo{volume}{291}, \bibinfo{pages}{438--446}.
\bibitem[{Bessler et~al.(2016)Bessler, Leonhardt and
  Wolff}]{bessler2016analyzing}
\bibinfo{author}{Bessler, W.}, \bibinfo{author}{Leonhardt, A.},
  \bibinfo{author}{Wolff, D.}, \bibinfo{year}{2016}.
\newblock \bibinfo{title}{Analyzing hedging strategies for fixed income
  portfolios: A bayesian approach for model selection}.
\newblock \bibinfo{journal}{International Review of Financial Analysis}
  \bibinfo{volume}{46}, \bibinfo{pages}{239--256}.
\bibitem[{Britten-Jones(1999)}]{Britten1999}
\bibinfo{author}{Britten-Jones, M.}, \bibinfo{year}{1999}.
\newblock \bibinfo{title}{The sampling error in estimates of mean-variance
  efficient portfolio weights}.
\newblock \bibinfo{journal}{The Journal of Finance} \bibinfo{volume}{54},
  \bibinfo{pages}{655--671}.
\bibitem[{Carr et~al.(2001)Carr, Geman and Madan}]{carr2001pricing}
\bibinfo{author}{Carr, P.}, \bibinfo{author}{Geman, H.},
  \bibinfo{author}{Madan, D.B.}, \bibinfo{year}{2001}.
\newblock \bibinfo{title}{Pricing and hedging in incomplete markets}.
\newblock \bibinfo{journal}{Journal of financial economics}
  \bibinfo{volume}{62}, \bibinfo{pages}{131--167}.
\bibitem[{Castagnoli et~al.(2021)Castagnoli, Cattelan, Maccheroni, Tebaldi and
  Wang}]{Castagnoli2021}
\bibinfo{author}{Castagnoli, E.}, \bibinfo{author}{Cattelan, G.},
  \bibinfo{author}{Maccheroni, F.}, \bibinfo{author}{Tebaldi, C.},
  \bibinfo{author}{Wang, R.}, \bibinfo{year}{2021}.
\newblock \bibinfo{title}{Star-shaped risk measures}.
\newblock \URLprefix \url{https://arxiv.org/abs/2103.15790},
  \DOIprefix\doi{10.48550/ARXIV.2103.15790}.
\bibitem[{Cherny and Madan(2009)}]{Cherny2009}
\bibinfo{author}{Cherny, A.}, \bibinfo{author}{Madan, D.},
  \bibinfo{year}{2009}.
\newblock \bibinfo{title}{New measures for performance evaluation}.
\newblock \bibinfo{journal}{Review of Financial Studies} \bibinfo{volume}{22},
  \bibinfo{pages}{2371--2406}.
\bibitem[{Cont et~al.(2010)Cont, Deguest and Scandolo}]{Cont2010}
\bibinfo{author}{Cont, R.}, \bibinfo{author}{Deguest, R.},
  \bibinfo{author}{Scandolo, G.}, \bibinfo{year}{2010}.
\newblock \bibinfo{title}{Robustness and sensitivity analysis of risk
  measurement procedures}.
\newblock \bibinfo{journal}{Quantitative finance} \bibinfo{volume}{10},
  \bibinfo{pages}{593--606}.
\bibitem[{Daouia et~al.(2019)Daouia, Gijbels and Stupfler}]{Daouia2019}
\bibinfo{author}{Daouia, A.}, \bibinfo{author}{Gijbels, I.},
  \bibinfo{author}{Stupfler, G.}, \bibinfo{year}{2019}.
\newblock \bibinfo{title}{Extremiles: A new perspective on asymmetric least
  squares}.
\newblock \bibinfo{journal}{Journal of the American Statistical Association}
  \bibinfo{volume}{114}, \bibinfo{pages}{1366--1381}.
\bibitem[{Daouia et~al.(2021)Daouia, Gijbels and Stupfler}]{Daouia2021}
\bibinfo{author}{Daouia, A.}, \bibinfo{author}{Gijbels, I.},
  \bibinfo{author}{Stupfler, G.}, \bibinfo{year}{2021}.
\newblock \bibinfo{title}{Extremile regression}.
\newblock \bibinfo{journal}{Journal of the American Statistical Association} ,
  \bibinfo{pages}{1--8}.
\bibitem[{Delbaen(2012)}]{Delbaen2012}
\bibinfo{author}{Delbaen, F.}, \bibinfo{year}{2012}.
\newblock \bibinfo{title}{Monetary utility functions}.
\newblock \bibinfo{publisher}{Osaka University Press}.
\bibitem[{Dennis~Jr and Welsch(1978)}]{dennis1978}
\bibinfo{author}{Dennis~Jr, J.E.}, \bibinfo{author}{Welsch, R.E.},
  \bibinfo{year}{1978}.
\newblock \bibinfo{title}{Techniques for nonlinear least squares and robust
  regression}.
\newblock \bibinfo{journal}{Communications in Statistics-simulation and
  Computation} \bibinfo{volume}{7}, \bibinfo{pages}{345--359}.
\bibitem[{Embrechts et~al.(2021)Embrechts, Mao, Wang and Wang}]{Embrechts2021}
\bibinfo{author}{Embrechts, P.}, \bibinfo{author}{Mao, T.},
  \bibinfo{author}{Wang, Q.}, \bibinfo{author}{Wang, R.}, \bibinfo{year}{2021}.
\newblock \bibinfo{title}{Bayes risk, elicitability, and the {Expected
  Shortfall}}.
\newblock \bibinfo{journal}{Mathematical Finance} \bibinfo{volume}{31},
  \bibinfo{pages}{1190--1217}.
\bibitem[{Fan et~al.(2012)Fan, Zhang and Yu}]{Fan2012}
\bibinfo{author}{Fan, J.}, \bibinfo{author}{Zhang, J.}, \bibinfo{author}{Yu,
  K.}, \bibinfo{year}{2012}.
\newblock \bibinfo{title}{Vast portfolio selection with gross-exposure
  constraints}.
\newblock \bibinfo{journal}{Journal of the American Statistical Association}
  \bibinfo{volume}{107}, \bibinfo{pages}{592--606}.
\bibitem[{Fischer(2003)}]{fischer2003risk}
\bibinfo{author}{Fischer, T.}, \bibinfo{year}{2003}.
\newblock \bibinfo{title}{Risk capital allocation by coherent risk measures
  based on one-sided moments}.
\newblock \bibinfo{journal}{Insurance: Mathematics and Economics}
  \bibinfo{volume}{32}, \bibinfo{pages}{135--146}.
\bibitem[{Fissler and Ziegel(2016)}]{fissler2016higher}
\bibinfo{author}{Fissler, T.}, \bibinfo{author}{Ziegel, J.F.},
  \bibinfo{year}{2016}.
\newblock \bibinfo{title}{{Higher order elicitability and Osband's principle}}.
\newblock \bibinfo{journal}{The Annals of Statistics} \bibinfo{volume}{44},
  \bibinfo{pages}{1680--1707}.
\bibitem[{Fissler and Ziegel(2021)}]{fissler2021}
\bibinfo{author}{Fissler, T.}, \bibinfo{author}{Ziegel, J.F.},
  \bibinfo{year}{2021}.
\newblock \bibinfo{title}{On the elicitability of range value at risk}.
\newblock \bibinfo{journal}{Statistics \& Risk Modeling} \bibinfo{volume}{38},
  \bibinfo{pages}{25--46}.
\bibitem[{F\"{o}llmer and Knispel(2013)}]{Follmer2013}
\bibinfo{author}{F\"{o}llmer, H.}, \bibinfo{author}{Knispel, T.},
  \bibinfo{year}{2013}.
\newblock \bibinfo{title}{Convex risk measures: Basic facts, law-invariance and
  beyond, asymptotics for large portfolios}, in: \bibinfo{editor}{MacLean, L.},
  \bibinfo{editor}{Ziemba, W.} (Eds.), \bibinfo{booktitle}{Handbook of the
  Fundamentals of Financial Decision Making}. \bibinfo{publisher}{World
  Scientific}, pp. \bibinfo{pages}{507--554}.
\bibitem[{F{\"o}llmer and Schied(2002)}]{Follmer2002}
\bibinfo{author}{F{\"o}llmer, H.}, \bibinfo{author}{Schied, A.},
  \bibinfo{year}{2002}.
\newblock \bibinfo{title}{Convex measures of risk and trading constraints}.
\newblock \bibinfo{journal}{Finance and stochastics} \bibinfo{volume}{6},
  \bibinfo{pages}{429--447}.
\bibitem[{Follmer and Schied(2016)}]{Follmer2016}
\bibinfo{author}{Follmer, H.}, \bibinfo{author}{Schied, A.},
  \bibinfo{year}{2016}.
\newblock \bibinfo{title}{Stochastic finance: an introduction in discrete
  time}.
\newblock \bibinfo{publisher}{Walter de Gruyter GmbH}.
\bibitem[{F{\"o}llmer and Weber(2015)}]{Follmer2015}
\bibinfo{author}{F{\"o}llmer, H.}, \bibinfo{author}{Weber, S.},
  \bibinfo{year}{2015}.
\newblock \bibinfo{title}{The axiomatic approach to risk measures for capital
  determination}.
\newblock \bibinfo{journal}{Annual Review of Financial Economics}
  \bibinfo{volume}{7}, \bibinfo{pages}{301--337}.
\bibitem[{Frey and Pohlmeier(2016)}]{Frey2016}
\bibinfo{author}{Frey, C.}, \bibinfo{author}{Pohlmeier, W.},
  \bibinfo{year}{2016}.
\newblock \bibinfo{title}{Bayesian shrinkage of portfolio weights}.
\newblock \bibinfo{journal}{Available at SSRN 2730475} .
\bibitem[{Frittelli and Gianin(2002)}]{Frittelli2002}
\bibinfo{author}{Frittelli, M.}, \bibinfo{author}{Gianin, E.R.},
  \bibinfo{year}{2002}.
\newblock \bibinfo{title}{Putting order in risk measures}.
\newblock \bibinfo{journal}{Journal of Banking \& Finance}
  \bibinfo{volume}{26}, \bibinfo{pages}{1473--1486}.
\bibitem[{Gerlach et~al.(2017)Gerlach, Walpole and Wang}]{gerlach2017semi}
\bibinfo{author}{Gerlach, R.}, \bibinfo{author}{Walpole, D.},
  \bibinfo{author}{Wang, C.}, \bibinfo{year}{2017}.
\newblock \bibinfo{title}{Semi-parametric {Bayesian} tail risk forecasting
  incorporating realized measures of volatility}.
\newblock \bibinfo{journal}{Quantitative Finance} \bibinfo{volume}{17},
  \bibinfo{pages}{199--215}.
\bibitem[{Gianin and Sgarra(2013)}]{Gianin2013}
\bibinfo{author}{Gianin, E.R.}, \bibinfo{author}{Sgarra, C.},
  \bibinfo{year}{2013}.
\newblock \bibinfo{title}{Acceptability indexes via 'g-expectations': An
  application to liquidity risk}.
\newblock \bibinfo{journal}{Mathematics and Financial Economics}
  \bibinfo{volume}{7}, \bibinfo{pages}{457–475}.
\bibitem[{Gneiting(2011)}]{gneiting2011making}
\bibinfo{author}{Gneiting, T.}, \bibinfo{year}{2011}.
\newblock \bibinfo{title}{Making and evaluating point forecasts}.
\newblock \bibinfo{journal}{Journal of the American Statistical Association}
  \bibinfo{volume}{106}, \bibinfo{pages}{746--762}.
\bibitem[{Grechuk et~al.(2009)Grechuk, Molyboha and Zabarankin}]{Grechuk2009}
\bibinfo{author}{Grechuk, B.}, \bibinfo{author}{Molyboha, A.},
  \bibinfo{author}{Zabarankin, M.}, \bibinfo{year}{2009}.
\newblock \bibinfo{title}{{Maximum Entropy Principle with General Deviation
  Measures}}.
\newblock \bibinfo{journal}{Mathematics of Operations Research}
  \bibinfo{volume}{34}, \bibinfo{pages}{445--467}.
\bibitem[{Guillen et~al.(2021)Guillen, Berm{\'u}dez and
  Pitarque}]{guillen2021joint}
\bibinfo{author}{Guillen, M.}, \bibinfo{author}{Berm{\'u}dez, L.},
  \bibinfo{author}{Pitarque, A.}, \bibinfo{year}{2021}.
\newblock \bibinfo{title}{Joint generalized quantile and conditional tail
  expectation regression for insurance risk analysis}.
\newblock \bibinfo{journal}{Insurance: Mathematics and Economics}
  \bibinfo{volume}{99}, \bibinfo{pages}{1--8}.
\bibitem[{Halkos and Tsirivis(2019)}]{halkos2019energy}
\bibinfo{author}{Halkos, G.E.}, \bibinfo{author}{Tsirivis, A.S.},
  \bibinfo{year}{2019}.
\newblock \bibinfo{title}{Energy commodities: A review of optimal hedging
  strategies}.
\newblock \bibinfo{journal}{Energies} \bibinfo{volume}{12},
  \bibinfo{pages}{3979}.
\bibitem[{Herdegen and Khan(2022)}]{Herdegen2021}
\bibinfo{author}{Herdegen, M.}, \bibinfo{author}{Khan, N.},
  \bibinfo{year}{2022}.
\newblock \bibinfo{title}{Sensitivity to large losses and $\rho$-arbitrage for
  convex risk measures}.
\newblock \URLprefix \url{https://arxiv.org/abs/2202.07610},
  \DOIprefix\doi{10.48550/ARXIV.2202.07610}.
\bibitem[{Huang and Guo(2013)}]{huang2013optimal}
\bibinfo{author}{Huang, S.F.}, \bibinfo{author}{Guo, M.}, \bibinfo{year}{2013}.
\newblock \bibinfo{title}{An optimal multi-step quadratic risk-adjusted hedging
  strategy}.
\newblock \bibinfo{journal}{Journal of the Korean Statistical Society}
  \bibinfo{volume}{42}, \bibinfo{pages}{37--49}.
\bibitem[{Huber(1992)}]{huber1992}
\bibinfo{author}{Huber, P.J.}, \bibinfo{year}{1992}.
\newblock \bibinfo{title}{Robust estimation of a location parameter}, in:
  \bibinfo{booktitle}{Breakthroughs in statistics}.
  \bibinfo{publisher}{Springer}, pp. \bibinfo{pages}{492--518}.
\bibitem[{Kaina and R{\"u}schendorf(2009)}]{Kaina2009}
\bibinfo{author}{Kaina, M.}, \bibinfo{author}{R{\"u}schendorf, L.},
  \bibinfo{year}{2009}.
\newblock \bibinfo{title}{On convex risk measures on lp-spaces}.
\newblock \bibinfo{journal}{Mathematical Methods of Operations Research}
  \bibinfo{volume}{69}, \bibinfo{pages}{475--495}.
\bibitem[{Kempf and Memmel(2006)}]{Kempf2006}
\bibinfo{author}{Kempf, A.}, \bibinfo{author}{Memmel, C.},
  \bibinfo{year}{2006}.
\newblock \bibinfo{title}{Estimating the global minimum variance portfolio}.
\newblock \bibinfo{journal}{Schmalenbach Business Review} \bibinfo{volume}{58},
  \bibinfo{pages}{332--348}.
\bibitem[{Koenker(2005)}]{koenker2005quantile}
\bibinfo{author}{Koenker, R.}, \bibinfo{year}{2005}.
\newblock \bibinfo{title}{Quantile regression}.
\newblock \bibinfo{publisher}{Cambridge University Press New York}.
\bibitem[{Koenker and Bassett(1978)}]{Koenker1978}
\bibinfo{author}{Koenker, R.}, \bibinfo{author}{Bassett, G.},
  \bibinfo{year}{1978}.
\newblock \bibinfo{title}{Regression quantiles}.
\newblock \bibinfo{journal}{Econometrica} \bibinfo{volume}{46},
  \bibinfo{pages}{33--50}.
\bibitem[{Kou and Peng(2016)}]{Kou2016}
\bibinfo{author}{Kou, S.}, \bibinfo{author}{Peng, X.}, \bibinfo{year}{2016}.
\newblock \bibinfo{title}{On the measurement of economic tail risk}.
\newblock \bibinfo{journal}{Operations Research} \bibinfo{volume}{64},
  \bibinfo{pages}{1056--1072}.
\bibitem[{Landsman and Shushi(2022)}]{landsman2022}
\bibinfo{author}{Landsman, Z.}, \bibinfo{author}{Shushi, T.},
  \bibinfo{year}{2022}.
\newblock \bibinfo{title}{The location of a minimum variance squared distance
  functional}.
\newblock \bibinfo{journal}{Insurance: Mathematics and Economics}
  \bibinfo{volume}{105}, \bibinfo{pages}{64--78}.
\bibitem[{Li(2015)}]{Li2015}
\bibinfo{author}{Li, J.}, \bibinfo{year}{2015}.
\newblock \bibinfo{title}{{Sparse and Stable Portfolio Selection With Parameter
  Uncertainty}}.
\newblock \bibinfo{journal}{Journal of Business \& Economic Statistics}
  \bibinfo{volume}{33}, \bibinfo{pages}{381--392}.
\newblock \DOIprefix\doi{10.1080/07350015.2014.954}.
\bibitem[{Liebrich(2021)}]{Liebrich2021}
\bibinfo{author}{Liebrich, F.B.}, \bibinfo{year}{2021}.
\newblock \bibinfo{title}{Risk sharing under heterogeneous beliefs without
  convexity}.
\newblock \URLprefix \url{https://arxiv.org/abs/2108.05791},
  \DOIprefix\doi{10.48550/ARXIV.2108.05791}.
\bibitem[{Mao and Cai(2018)}]{Mao2018}
\bibinfo{author}{Mao, T.}, \bibinfo{author}{Cai, J.}, \bibinfo{year}{2018}.
\newblock \bibinfo{title}{Risk measures based on behavioural economics theory}.
\newblock \bibinfo{journal}{Finance and Stochastics} \bibinfo{volume}{22},
  \bibinfo{pages}{367--393}.
\bibitem[{Markowitz(1952)}]{markowitz1952portfolio}
\bibinfo{author}{Markowitz, H.}, \bibinfo{year}{1952}.
\newblock \bibinfo{title}{Portfolio selection}.
\newblock \bibinfo{journal}{The Journal of Finance} \bibinfo{volume}{7},
  \bibinfo{pages}{77--91}.
\bibitem[{Moresco and Righi(2022)}]{Moresco2021}
\bibinfo{author}{Moresco, M.R.}, \bibinfo{author}{Righi, M.B.},
  \bibinfo{year}{2022}.
\newblock \bibinfo{title}{On the link between monetary and star-shaped risk
  measures}.
\newblock \bibinfo{journal}{Statistics \& Probability Letters} ,
  \bibinfo{pages}{109345}.
\bibitem[{Newey and Powell(1987)}]{Newey1987}
\bibinfo{author}{Newey, W.K.}, \bibinfo{author}{Powell, J.L.},
  \bibinfo{year}{1987}.
\newblock \bibinfo{title}{Asymmetric least squares estimation and testing}.
\newblock \bibinfo{journal}{Econometrica: Journal of the Econometric Society} ,
  \bibinfo{pages}{819--847}.
\bibitem[{Ogryczak and Ruszczy{\'n}ski(1999)}]{ogryczak1999stochastic}
\bibinfo{author}{Ogryczak, W.}, \bibinfo{author}{Ruszczy{\'n}ski, A.},
  \bibinfo{year}{1999}.
\newblock \bibinfo{title}{From stochastic dominance to mean-risk models:
  Semideviations as risk measures}.
\newblock \bibinfo{journal}{European Journal of Operational Research}
  \bibinfo{volume}{116}, \bibinfo{pages}{33--50}.
\bibitem[{Pflug and R\"{o}misch(2007)}]{Pflug2007}
\bibinfo{author}{Pflug, G.}, \bibinfo{author}{R\"{o}misch, W.},
  \bibinfo{year}{2007}.
\newblock \bibinfo{title}{Modeling, Measuring and Managing Risk}.
\newblock \bibinfo{edition}{1} ed., \bibinfo{publisher}{World Scientific}.
\bibitem[{Pflug(2006)}]{pflug2006subdifferential}
\bibinfo{author}{Pflug, G.C.}, \bibinfo{year}{2006}.
\newblock \bibinfo{title}{Subdifferential representations of risk measures}.
\newblock \bibinfo{journal}{Mathematical Programming} \bibinfo{volume}{108},
  \bibinfo{pages}{339--354}.
\bibitem[{Righi(2018)}]{Righi2018b}
\bibinfo{author}{Righi, M.B.}, \bibinfo{year}{2018}.
\newblock \bibinfo{title}{A theory for combinations of risk measures}.
\newblock \URLprefix \url{https://arxiv.org/abs/1807.01977},
  \DOIprefix\doi{10.48550/ARXIV.1807.01977}.
\bibitem[{Righi(2019)}]{righi2019composition}
\bibinfo{author}{Righi, M.B.}, \bibinfo{year}{2019}.
\newblock \bibinfo{title}{A composition between risk and deviation measures}.
\newblock \bibinfo{journal}{Annals of Operations Research}
  \bibinfo{volume}{282}, \bibinfo{pages}{299--313}.
\bibitem[{Righi(2021)}]{Righi2021b}
\bibinfo{author}{Righi, M.B.}, \bibinfo{year}{2021}.
\newblock \bibinfo{title}{Star-shaped acceptability indexes}.
\newblock \URLprefix \url{https://arxiv.org/abs/2110.08630},
  \DOIprefix\doi{10.48550/ARXIV.2110.08630}.
\bibitem[{Righi and Borenstein(2018)}]{Righi2018d}
\bibinfo{author}{Righi, M.B.}, \bibinfo{author}{Borenstein, D.},
  \bibinfo{year}{2018}.
\newblock \bibinfo{title}{A simulation comparison of risk measures for
  portfolio optimization}.
\newblock \bibinfo{journal}{Finance Research Letters} \bibinfo{volume}{24},
  \bibinfo{pages}{105--112}.
\bibitem[{Righi and Ceretta(2016)}]{Righi2016}
\bibinfo{author}{Righi, M.B.}, \bibinfo{author}{Ceretta, P.S.},
  \bibinfo{year}{2016}.
\newblock \bibinfo{title}{Shortfall deviation risk: an alternative for risk
  measurement}.
\newblock \bibinfo{journal}{Journal of Risk} \bibinfo{volume}{19},
  \bibinfo{pages}{81--116}.
\bibitem[{Righi et~al.(2020)Righi, Müller and Moresco}]{Righi2020}
\bibinfo{author}{Righi, M.B.}, \bibinfo{author}{Müller, F.M.},
  \bibinfo{author}{Moresco, M.R.}, \bibinfo{year}{2020}.
\newblock \bibinfo{title}{On a robust risk measurement approach for capital
  determination errors minimization}.
\newblock \bibinfo{journal}{Insurance: Mathematics and Economics}
  \bibinfo{volume}{95}, \bibinfo{pages}{199--211}.
\bibitem[{Rockafellar and Uryasev(2013)}]{Rockafellar2013}
\bibinfo{author}{Rockafellar, R.}, \bibinfo{author}{Uryasev, S.},
  \bibinfo{year}{2013}.
\newblock \bibinfo{title}{{The fundamental risk quadrangle in risk management,
  optimization and statistical estimation}}.
\newblock \bibinfo{journal}{Surveys in Operations Research and Management
  Science} \bibinfo{volume}{18}, \bibinfo{pages}{33--53}.
\bibitem[{Rockafellar et~al.(2006)Rockafellar, Uryasev and
  Zabarankin}]{rockafellar2006generalized}
\bibinfo{author}{Rockafellar, R.T.}, \bibinfo{author}{Uryasev, S.},
  \bibinfo{author}{Zabarankin, M.}, \bibinfo{year}{2006}.
\newblock \bibinfo{title}{Generalized deviations in risk analysis}.
\newblock \bibinfo{journal}{Finance and Stochastics} \bibinfo{volume}{10},
  \bibinfo{pages}{51--74}.
\bibitem[{Rockafellar et~al.(2007)Rockafellar, Uryasev and
  Zabarankin}]{Rockafellar2007}
\bibinfo{author}{Rockafellar, R.T.}, \bibinfo{author}{Uryasev, S.},
  \bibinfo{author}{Zabarankin, M.}, \bibinfo{year}{2007}.
\newblock \bibinfo{title}{Equilibrium with investors using a diversity of
  deviation measures}.
\newblock \bibinfo{journal}{Journal of Banking \& Finance}
  \bibinfo{volume}{31}, \bibinfo{pages}{3251--3268}.
\bibitem[{R\"{u}schendorf(2013)}]{Ruschendor2013}
\bibinfo{author}{R\"{u}schendorf, L.}, \bibinfo{year}{2013}.
\newblock \bibinfo{title}{Mathematical Risk Analysis}.
\newblock \bibinfo{publisher}{Springer}.
\bibitem[{Ruszczy{\'{n}}ski and Shapiro(2006)}]{Ruszczynski2006}
\bibinfo{author}{Ruszczy{\'{n}}ski, A.}, \bibinfo{author}{Shapiro, A.},
  \bibinfo{year}{2006}.
\newblock \bibinfo{title}{Optimization of Risk Measures}.
  \bibinfo{publisher}{Springer}.
\newblock pp. \bibinfo{pages}{119--157}.
\bibitem[{Shapiro(2017)}]{Shapiro2017}
\bibinfo{author}{Shapiro, A.}, \bibinfo{year}{2017}.
\newblock \bibinfo{title}{Distributionally robust stochastic programming}.
\newblock \bibinfo{journal}{SIAM Journal on Optimization} \bibinfo{volume}{27},
  \bibinfo{pages}{2258--2275}.
\bibitem[{Sion(1958)}]{Sion1958}
\bibinfo{author}{Sion, M.}, \bibinfo{year}{1958}.
\newblock \bibinfo{title}{On general minimax theorems}.
\newblock \bibinfo{journal}{Pacific Journal of Mathematics}
  \bibinfo{volume}{8}, \bibinfo{pages}{171--176}.
\bibitem[{Sun and Ji(2017)}]{Sun2017}
\bibinfo{author}{Sun, C.}, \bibinfo{author}{Ji, S.}, \bibinfo{year}{2017}.
\newblock \bibinfo{title}{The least squares estimator of random variables under
  sublinear expectations}.
\newblock \bibinfo{journal}{Journal of Mathematical Analysis and Applications}
  \bibinfo{volume}{451}, \bibinfo{pages}{906 -- 923}.
\bibitem[{Wu et~al.(2023)Wu, Yang and Zhang}]{Wu2023}
\bibinfo{author}{Wu, Q.}, \bibinfo{author}{Yang, F.}, \bibinfo{author}{Zhang,
  P.}, \bibinfo{year}{2023}.
\newblock \bibinfo{title}{Conditional generalized quantiles based on expected
  utility model and equivalent characterization of properties}.
\newblock \URLprefix \url{https://arxiv.org/abs/2301.12420}.
\bibitem[{Zellner(1986)}]{zellner1986bayesian}
\bibinfo{author}{Zellner, A.}, \bibinfo{year}{1986}.
\newblock \bibinfo{title}{Bayesian estimation and prediction using asymmetric
  loss functions}.
\newblock \bibinfo{journal}{Journal of the American Statistical Association}
  \bibinfo{volume}{81}, \bibinfo{pages}{446--451}.
\bibitem[{Ziegel(2016)}]{Ziegel2016}
\bibinfo{author}{Ziegel, J.F.}, \bibinfo{year}{2016}.
\newblock \bibinfo{title}{Coherence and elicitability}.
\newblock \bibinfo{journal}{Mathematical Finance} \bibinfo{volume}{26},
  \bibinfo{pages}{901--918}.

\end{thebibliography}

\end{document}